\documentclass{sig-alternate}
\usepackage{algorithm}
\usepackage{algorithmic}

\def\vc#1{\mbox{\boldmath $#1$}}

\newcommand{\up}[2]{\smash{\lower-#2\hbox{#1}}}

\newtheorem{defi}{Definition}[section]
  \newtheorem{theo}{Theorem}[section]
  \newtheorem{prop}{Proposition}[section]
  \newtheorem{lem}{Lemma}[section]
\newtheorem{kei}{Corollary}[section]
\newtheorem{remark}{Remark}[section]

\makeatletter
    
    \@addtoreset{equation}{section}
  \makeatother

\begin{document}
%

%
%

\title{Asymptotic and Numerical Analysis of Multiserver Retrial Queue with Guard Channel for  Cellular Networks}

\numberofauthors{2} 
%
\author{
\alignauthor
Kazuki Kajiwara \\
       \affaddr{Department of Mathematical and Computing Sciences}\\
       \affaddr{Tokyo Institute of Technology}\\
       \affaddr{Ookayama, Meguro-ku, Tokyo 152-8552, Japan}\\
       \email{kajiwara1220@gmail.com}
\alignauthor
Tuan Phung-Duc\thanks{Corresponding author}\\
       \affaddr{Department of Mathematical and Computing Sciences}\\
       \affaddr{Tokyo Institute of Technology}\\
       \affaddr{Ookayama, Meguro-ku, Tokyo 152-8552, Japan}\\
       \email{tuan@is.titech.ac.jp}
}

\maketitle

\begin{abstract}
This paper considers a retrial queueing model for a base station in cellular networks where fresh calls and handover calls are available. Fresh calls are initiated from the cell of the base station. On the other hand, a handover call has been connecting to a base station and moves to another one. In order to keep the continuation of the communication, it is desired that  an available channel in the new base station is immediately assigned to the handover call. 
To this end, a channel is reserved as the guard channel for handover calls in base stations. Blocked fresh and handover calls join a virtual orbit and repeat their attempts in a later time. We assume that a base station can recognize retrial calls and give them the same priority as that of handover calls. We model a base station by a multiserver retrial queue with priority customers for which a level-dependent QBD process is formulated. We obtain Taylor series expansion for the nonzero elements of the rate matrices of the level-dependent QBD. Using the expansion results, we obtain an asymptotic upper bound for the joint stationary distribution of the number of busy channels and that of customers in the orbit. Furthermore, we derive an efficient numerical algorithm to calculate the joint stationary distribution.
\end{abstract}

\section{Introduction}
In this paper, we consider multiserver retrial queues with a guard channel for priority and retrial customers. Retrial queues are characterized by the fact that a blocked customer repeats its request after a random time. During 
retrial intervals, the customer is said to be in the orbit. This type of queueing models is widely used in modelling and performance analysis of communication and service systems, especially in cellular networks~\cite{abst,approximation_cellular,Do10,Artalejo_Lopez}. 
For instance, Tran-Gia and Mandjes~\cite{abst} report the influence of retrials on the performance of cellular networks using retrial queueing models. Marsan et al.~\cite{approximation_cellular} carries out a fixed point approximation analysis for retrial queueing models arising from 
cellular networks. Artalejo and Lopez-Herrero analyze a multiserver queue for cellular networks operating under a random environment using four-dimensional Markov chains. 
Do~\cite{Do10} investigates the model presented in~\cite{abst} by a fixed point approximation method based on the corresponding model with constant retrial rate.

In cellular networks, users may move from one cell to another cell. A handover call is a call that arrives at the current cell from an adjacent cell where it has been connecting with the base station of that cell. Thus, in order to keep the continuation of the communication, 
it is desired that an available channel is immediately assigned to a handover call upon its arrival. On the other hand, a fresh call is a call that is initiated from inside the cell of the base station. 
Therefore, from a quality of service (QoS) point of view, blocking of a handover call has more negative impact than that of a fresh call. 

Tran-Gia and Mandjes~\cite{abst} propose some multiserver retrial queues with fresh and handover calls and guard channels for a base station in cellular networks. In~\cite{abst}, the orbit size is assumed to be finite and the same priority is given for both retrial calls and fresh calls. 
It should be noted that the analysis is simplified by the finite assumption for the orbit size. 
In contrast to this, we consider in this paper a model with infinite orbit size where retrial calls and handover calls have higher priority than fresh calls. Although the base station needs to distinguish retrial calls and new calls (that arrive for the first time), this allows to reduce 
the number of retrials per customer and to improve the QoS. 

The analysis of multiserver retrial queues is challenging due to the fact that the underlying Markov chain is state nonhomogeneous because the retrial rate is proportional to the number of customers in the orbit. 
Thus, even for the fundamental model with one type of traffic and without guard channels, an analytical solution is available for only some special cases, i.e., one or two servers~\cite{bernoulli}. 

For models with both retrial and guard channels, although some numerical methods~\cite{abst,approximation_cellular,Do10,Artalejo_Lopez} have been presented, there is no analytical result available. 
This motivates us to consider a novel model with both retrials and a guard channel for which we explore both new analytical and numerical results. From the modelling point of view, the novelty is the priority given to retrial calls. 
To the best of our knowledge, this paper is the first to consider priority for retrial calls. 

In this paper, we consider only one guard channel. This assumption is not restrictive because i) the model with one guard channel is complex enough, ii) the analysis for arbitrary number of guard channels is straightforward 
and iii) in a cellular network context, one guard channel is enough to guarantee QoS~\cite {abst}. 
We formulate the queueing system using a level-dependent QBD process where the level and the phase are referred to as the number of calls in the orbit and that of busy channels, respectively. 
As is well known, the stationary distribution for multiserver retrial queue is analytically tractable for the case of one or two servers only~\cite{Retrial_queues}. Thus, it is difficult to get analytical insights into retrial queueing models. 
We refer to~\cite{Pearce89,phung1,bernoulli} for some efforts in finding analytical expressions for the joint stationary distribution.

 The stationary distribution of level-dependent QBDs can be expressed in terms of a sequence of rate matrices~\cite{Rdef}. Thus, we can characterize the stationary distribution through the sequence of rate matrices. 
The QBD process of our model possesses some special structure, i.e., only the last two rows are nonzero allowing us to get some insights into the structure of the stationary distribution. 
Liu and Zhao~\cite{analying_retrial} use this property to obtain upper and lower asymptotic bounds for the stationary distribution of the fundamental retrial model without guard channels. 
Liu et al.~\cite{Liub} further extend their analysis to the model with nonpersistent customers. 
Phung-Duc~\cite{Perturbation} presents a perturbation analysis for a multiserver retrial queues with two type of nonpersistent customers. 
In~\cite{Perturbation}, the author derives Taylor series expansion formulae for the nonzero elements of the rate matrices. 
The different point of our model in comparison with the above work is that the last two rows of the rate matrices are nonzero in our model while for those in~\cite{analying_retrial,Perturbation} only the last row is nonzero. 
This makes the analysis more complex and challenging. 

The main contribution of our paper is threefold. First, using a censoring technique and a perturbation method, we obtain Taylor series expansion for the rate matrices in terms of the number of customers in the orbit. 
Our formula is general in the sense that we can obtain the expansion with arbitrary number of terms, what was not reported in Liu and Zhao~\cite{analying_retrial}. Second, using this result we obtain an asymptotic upper bound for the stationary distribution which 
is more challenging than~\cite{analying_retrial} and~\cite{Perturbation} due to the denseness of the rate matrices.  
Third, using the special structure of the rate matrix and a matrix continued fraction approach~\cite{matrix-conf}, we propose an efficient numerical algorithm extending that of Phung-Duc et al.~\cite{sutikeisan} for the rate matrices and then for the stationary distribution. 
The computational complexity of the algorithm is in the order of the number of channels. 

The rest of our paper is organized as follows. Section 2 presents the model and some preliminary results on the QBD formulation. 
Section 3 is devoted to the presentation of Taylor series expansion for the rate matrices. In Section 4, we show the asymptotic upper bound for the joint stationary distribution while a numerical algorithm for the joint stationary distribution 
is presented in Section 5. Section 6 provides some numerical examples. Section 7 concludes our paper and presents some future directions.

\section{Model and Formulation}

\subsection{Model}
In this paper, we consider a queueing model with two types of customers (types 1 and 2).  
There are $c$ servers among them one server is designed as the reserved server which corresponds to the guard channel in cellular networks. 
Customers of type 1 (high priority) and type 2 (low priority) arrive at the system according to the Poisson processes with rate $\lambda_1$ and 
$\lambda_2$, respectively. Customers of type 1 can use all $c$ servers while those of type 2 cannot use the guard server. Thus, if there are $c-1$ busy servers, 
the last server automatically becomes the guard server for customers of type 1. 
Customers of types 1 and 2 correspond to handover calls and fresh calls, respectively. 
Furthermore, we assume that a blocked call (both types 1 and 2) redials after some exponentially distributed time with mean $1/\mu$. 
Upon retrial, if there is an idle channel the call occupies it immediately, otherwise it enters the orbit again. Thus, a redial call has the same priority as that of a handover call. 
In this paper, we assume that the base station can distinguish redials calls so as to give them the same priority as of handover calls. 
As a result, we may expect that decreasing the number of retrials by a customer improves the QoS. 
Service times for both fresh calls and handover calls are assumed to follow the same exponential distribution with mean $1/\nu$.

\subsection{Level-dependent QBD process}
Let $C(t)$ and $N(t)$ denote the number of busy channels and the number of redial calls in the orbit at time $t$. 
Letting $X(t) = (C(t),N(t)) \ (t \geq 0)$, the bivariate process $\{X(t);t \geq 0 \}$ is a Markov chain in the state space $\mathcal{S} = \{ 0,1,\dots ,c \} \times  \mathbb{Z}_{+} $, 
where $ \mathbb{Z}_{+} = \{  0,1,2,\dots \} $. We assume that $\{X(t)\}$ is positive recurrent. 
The necessary and sufficient condition for the positive recurrence of $\{X(t)\}$ is given in the following lemma. 
\begin{lem}\label{lem:stationary}
$\{X(t)\}$ is positive recurrent if and only if 
\[
	\frac{\lambda}{c\nu} < 1,
\]
\end{lem}
where $\lambda = \lambda_1 + \lambda_2$.

\begin{proof}
The proof is presented in Appendix~\ref{app:A0}.
\end{proof}

It is easy to see that $\{X(t); t \geq 0\}$ is a level-dependent QBD process whose infinitesimal generator $\vc{Q}$ is given as follows.
\begin{eqnarray}
\vc{Q} =
\left( 
\begin{array}{ccccc}
\vc{Q}_1^{(0)}&\vc{Q}_0^{(0)}&\vc{O}&\vc{O}&\cdots \\
\vc{Q}_2^{(1)} & \vc{Q}_1^{(1)} & \vc{Q}_0^{(1)}&\vc{O}&\cdots \\
\vc{O}&\vc{Q}_2^{(2)} & \vc{Q}_1^{(2)} & \vc{Q}_0^{(2)}&\cdots \\
\vc{O}&\vc{O}&\vc{Q}_2^{(3)} & \vc{Q}_1^{(3)} &\cdots \\
\vdots&\vdots&\vdots&\vdots&\ddots \\
\end{array} 
\right), \label{eq:Q}
\end{eqnarray}
where $\vc{O}$ is the zero matrix with appropriate dimension and  
$\{\vc{Q}_0^{(n)},\ \vc{Q}_1^{(n)}; n \in \mathbb{Z}_{+}\}$ ,\ 
$\{\vc{Q}_2^{(n)}; n \in \mathbb{N}\}$ are square matrices of size $c+1$ given as follows. 
\begin{eqnarray*}
\vc{Q}_{0}^{(n)} & =
\left( 
\begin{array}{ccccc}
0&\cdots &0&0&0 \\
\vdots & \ddots & \vdots&\vdots&\vdots \\
0&\cdots &0&0&0 \\
0& \cdots &0&\lambda_{2}&0 \\
0&\cdots&0&0&\lambda\\
\end{array} 
\right)
\\
\vc{Q}_{2}^{(n)} & = \left( 
\begin{array}{ccccc}
0 & n\mu&0&\cdots&0 \\
0&0&n\mu&\ddots& 0\\
0 &&\ddots&\ddots &0  \\
\vdots&&&0&n\mu\\
0&\cdots&\cdots&0&0\\
\end{array} 
\right),\ 
\end{eqnarray*}
\begin{eqnarray*}
\vc{Q}_{1}^{(n)} = 
\left(
\begin{array}{cccccc}
b_{0}^{(n)}&\lambda&0&\cdots&\cdots&0 \\
\nu&b_{1}^{(n)}&\lambda&\ddots&&\vdots \\
0&2\nu&b_{2}^{(n)}&\ddots&\ddots&\vdots \\
\vdots&\ddots&\ddots&\ddots&\lambda&0 \\
\vdots&&\ddots&\ddots&b_{c-1}^{(n)}&\lambda_{1} \\
0&\cdots&\cdots&0&c\nu&b_{c}^{(n)} \\
\end{array}
\right),
\end{eqnarray*} 
where $\mathbb{N} = \{1,2,\dots\}$, $b_{i}^{(n)} =  - (\lambda +i\nu + n\mu(1 -\delta_{i,c})) \ (i=0,1,2,\dots,c$) and $\delta_{i,c}$ is the Kronecker symbol. 
Let $\pi_{i,n}$ denote the stationary probability that there are $i$ busy servers and $n$ redial calls in the orbit, i.e., 
\begin{eqnarray}
\pi_{i,n} =\lim_{t \to \infty} {\rm Pr}(C(t)=i,N(t)=n), \label{eq:pi_teigi}
\end{eqnarray} 
for $i=0,1,\dots,c$ and  $n \in \mathbb{Z}_{+}$. 
Furthermore, let 
\[ \vc{\pi}_{n} = (\pi_{0,n}, \pi_{1,n},\dots,\pi_{c,n}), \quad \vc{\pi}=(\vc{\pi}_0, \vc{\pi}_1,\dots). \] 
We have 
\begin{eqnarray}
\vc{\pi}_{0} \vc{Q}_{1}^{(0)}+\vc{\pi}_{1} \vc{Q}_{2}^{(1)} &=& \vc{0}, \ n=0, \label{eq:n=0} \\
\vc{\pi}_{n-1} \vc{Q}_{0}^{(n-1)} +\vc{\pi}_{n} \vc{Q}_{1}^{(n)}+\vc{\pi}_{n+1} \vc{Q}_{2}^{(n+1)} &=& \vc{0}, \ n\in \mathbb{N},  \qquad \label{eq:pi_Q}\\
 \vc{\pi e} &=& 1, \label{eq:pi*e=1}
\end{eqnarray}
where $\vc{e}$  and $\vc{0}$ are vectors with appropriate dimensions with all 1 elements and all zero elements, respectively. 
It is established in~\cite{Rdef} that the solution of (\ref{eq:n=0}), (\ref{eq:pi_Q}) and (\ref{eq:pi*e=1}) is given by
\[
\vc{\pi}_{n} = \vc{\pi}_{n-1} \vc{R}^{(n)} ,\qquad n \in \mathbb{N}, 
\]
where $\{ \vc{R}^{(n)};n \in \mathbb{N}\} $ is the minimal nonnegative solution of 
\begin{eqnarray}
\vc{\vc{Q}}_{0}^{(n-1)} +\vc{R}^{(n)} \vc{Q}_{1}^{(n)}+\vc{R}^{(n)}\vc{R}^{(n+1)} \vc{Q}_{2}^{(n+1)} = \vc{O}, \ n \in \mathbb{N}. \label{eq:Rteigi}
\end{eqnarray}
Furthermore, $\vc{\pi}_0$ is determined by   
\begin{eqnarray*}
\vc{\pi}_{0}( \vc{Q}_{1}^{(0)}+\vc{R}^{(1)} \vc{Q}_{2}^{(1)}) &=& \vc{0}, \\
\vc{\pi}_{0} (\vc{I} + \vc{R}^{(1)}  + \vc{R}^{(1)} \vc{R}^{(2)} + \dots ) \vc{e}&=&1.
\end{eqnarray*}
Thus the problem of finding the stationary distribution is equivalent to that of obtaining the rate matrices. 
However, the rate matrices do not have closed form in general leading to an algorithmic approach for numerical calculation. 
To this end, we present the following three lemmas. 
\begin{lem} \label{prop:1}
{\rm (Proposition 1 in~\cite{Perturbation})
Let $\mathcal{M}$ denote the set of square matrices of size $c+1$. Furthermore, let $R_n:\mathcal{M} \to \mathcal{M}$ denote the following function. 
\[
R_n(\vc{X}) =  - \vc{Q}_{0}^{(n-1)} ( \vc{Q}_{1}^{(n)} + \vc{X} \vc{Q}_{2}^{(n+1)})^{-1}
,\qquad n\in \mathbb{N}.
\]
It is easy to see that $\{\vc{R}^{(n)}; n \in \mathbb{N} \}$ satisfies 
\[
\vc{R}^{(n)} = R_n(\vc{R}^{(n+1)})= R_n \circ R_{n+1} \circ R_{n+2} \circ \cdots
,\quad n \in \mathbb{N}, 
\]
where $f(g(\cdot))=f \circ g(\cdot)$.}
\end{lem}
\begin{lem}\label{prop:2}{\rm
(Proposition 2 in~\cite{Perturbation})
$\{\vc{R}^{(n)}_{k}; k \in\mathbb{Z}_{+} \}$ is defined by the following recursive formulae.  
\begin{eqnarray*}
\vc{R}_0^{(n)} &=& \vc{O} ,\qquad k=0,\\
\vc{R}_{k}^{(n)} &=& R_n(\vc{R}_{k-1}^{(n+1)}) \\
& &  \vdots \\
& = & R_{n} \circ R_{n+1} \circ   \dots  \circ  {R}_{n+k-1}(\vc{O}), \qquad n,k \in \mathbb{N}.
\end{eqnarray*}
We have
\begin{equation*}
\lim_{k \to \infty} \vc{R}^{(n)}_{k} = \vc{R}^{(n)},\qquad n \in \mathbb{N}.\label{eq:Rnk}
\end{equation*}}
\end{lem}
Lemmas~\ref{prop:1} and \ref{prop:2} allow deriving a numerical algorithm for calculating the rate matrices. 
They also show that the rate matrices are matrix continued fractions. 
However it is difficult to get insights into the rate matrices using the matrix continued fraction representation. 

In this paper, we show a Taylor series expansion of $\vc{R}^{(n)}$ in terms of $1/n$. 
It follows from Lemma~\ref{prop:1} that the first $c-1$ rows of $\vc{R}^{(n)}$ are zero.
Let $\vc{r}^{(0,n)}$ and $\vc{r}^{(1,n)}$ denote the $c$-th and $(c+1)$-th rows of $\vc{R}^{(n)}$, i.e.,  
\begin{eqnarray*}
\vc{r}^{(0,n)} & = & \left( r_{0}^{(0,n)}, r_{1}^{(0,n)}, \dots, r_{c}^{(0,n)} \right), \\ 
\vc{r}^{(1,n)} & = & \left( r_{0}^{(1,n)}, r_{1}^{(1,n)}, \dots, r_{c}^{(1,n)} \right). 
\end{eqnarray*}
Comparing the last two rows in both sides of (\ref{eq:Rteigi}) yields
\begin{eqnarray}
b_{0}^{(n)}r_{0}^{(0,n)} + \nu r_{1}^{(0,n)} & = & 0, \label{eq:r00}  \\
\lambda r_{i-1}^{(0,n)} + b_{i}^{(n)} r_{i}^{(0,n)} +(i+1)\nu  r_{i+1}^{(0,n)} + \tilde{r}_{i}^{(0,n)} & = & 0, \ \label{eq:r01} \\
\qquad i=1,2,\ldots,c-2   & & \nonumber \\
\lambda r_{c-2}^{(0,n)} + b_{c-1}^{(n)} r_{c-1}^{(0,n)} +c\nu  r_{c}^{(0,n)} +  \tilde{r}_{c-1}^{(0,n)} & = & -\lambda_{2}, \label{eq:r02}\ \qquad \\
\lambda_{1}r_{c-1}^{(0,n)}+ b_{c}^{(n)}r_{c}^{(0,n)}+ \tilde{r}_{c}^{(0,n)} & = & 0,	\label{eq:r03} \\  
b_{0}^{(n)}r_{0}^{(1,n)} + \nu r_{1}^{(1,n)}&=&0,  \label{eq:r10}  \qquad \\
\lambda r_{i-1}^{(1,n)} + b_{i}^{(n)} r_{i}^{(1,n)} +(i+1)\nu  r_{i+1}^{(1,n)} + \tilde{r}_{i}^{(1,n)} &=&  0,   \\
\qquad                i=1,2,\ldots,c-1,  \label{eq:r11}   & &    \nonumber \\
\lambda_{1}r_{c-1}^{(1,n)}+ b_{c}^{(n)}r_{c}^{(1,n)}+ \tilde{r}_{c}^{(1,n)}&=& -\lambda,  \qquad \label{eq:r12}
\end{eqnarray}
where 
\begin{eqnarray*}
\tilde{r}_{i}^{(0,n)} & = & (n+1)\mu \left( r_{c-1}^{(0,n)}r_{i-1}^{(0,n+1)}+r_{c}^{(0,n)}r_{i-1}^{(1,n+1)}\right), \\
\tilde{r}_{i}^{(1,n)} & = & (n+1)\mu \left( r_{c-1}^{(1,n)}r_{i-1}^{(0,n+1)}+r_{c}^{(1,n)}r_{i-1}^{(1,n+1)}\right).
\end{eqnarray*}

\begin{lem} \label{lem:2-3}
{\rm (Proposition 3 in \cite{Perturbation})}
We have
\begin{eqnarray}
( \vc{Q}_{2}^{(n-1)} + \vc{Q}_{1}^{(n-1)}+\vc{R}^{(n)} \vc{Q}_{2}^{(n)} ) \vc{e} = \vc{0} ,\qquad n \in \mathbb{N}.\   \label{eq:5} 
\end{eqnarray}
Comparing the last two elements in both sides of (\ref{eq:5}) yields,  
\begin{eqnarray}
 \sum_{i=0}^{c-1} r_{i}^{(0,n)} & = & \frac{\lambda_2}{n\mu},  \label{sum1:eq} \\
\sum_{i=0}^{c-1} r_i^{(1,n)} & =  & \frac{\lambda}{n\mu} ,\qquad n \in \mathbb{N}.   \label{sum2:eq}
\end{eqnarray}
\end{lem}
%

%
\begin{proof}
This proposition follows from the fact that the following matrix represents the infinitesimal generator of 
the ergodic Markov chain $\{X(t); t \geq 0 \}$ censored in levels $\{ l(i); i = 0,1,\dots, n-1\}$, where $l(i) = ((0,i),(1,i),\dots,(c,i))$.
\begin{eqnarray}
\lefteqn{ \vc{Q}^{\leq n-1} = } \nonumber \\
& & 
\left (
\begin{array}{llllll}
	\vc{Q}^{(0)}_1 \  & \vc{Q}^{(0)}_0  \  & \vc{O}  \  &  \cdots \ & \vc{O} \
	\\
	\vc{Q}^{(1)}_2 & \vc{Q}^{(1)}_1   & \vc{Q}^{(1)}_0 & \ddots &  \vc{O}
	\\
	\vc{O}   & \vc{Q}^{(2)}_2 & \vc{Q}^{(2)}_1 & \ddots &  \vdots
	\\
	\vdots   &  \vc{O}    & \ddots & \ddots  & \vc{O} 
	\\
	\vdots   & \ddots     & \ddots   & \vc{Q}^{(n-2)}_2  & \vc{Q}^{(n-2)}_0	\\
	\vc{O}   & \cdots     & \vc{O}   & \vc{Q}^{(n-1)}_2  & \widehat{\vc{Q}}^{(n-1)}
\end{array}
\right ), \label{sensored:chain}  \nonumber \qquad
\end{eqnarray}
where 
\[
	\widehat{\vc{Q}}^{(n-1)}  = \vc{Q}^{(n-1)}_1 + \vc{R}^{(n)} \vc{Q}^{(n)}_2.
\]
Therefore, 
\[
	(\vc{Q}^{(n-1)}_2 + \widehat{\vc{Q}}^{(n-1)}) \vc{e} = \vc{0}. 
\] 
By comparing the last elements of both sides, we obtain the announced result.
\end{proof}

\begin{kei}
We present explicit expressions for the rate matrices $\vc{R}^{(n)}$ for the case $c=2$. 
It follows from (\ref{eq:r00}) and (\ref{sum1:eq}) with $c=2$ that 
\[
	r^{(0,n)}_0 = \frac{\lambda_2 \nu}{ n\mu (\lambda + \nu + n \mu)}, \quad r^{(0,n)}_1 = \frac{\lambda_2 (\lambda + n \mu)}{n\mu (\lambda + \nu + n\mu)}.
\]
Similarly, combining (\ref{sum2:eq}) and (\ref{eq:r10}) with $c=2$ yields
\[
	r^{(1,n)}_0 = \frac{\lambda \nu}{ n\mu (\lambda + \nu + n \mu)}, \quad r^{(1,n)}_1 = \frac{\lambda (\lambda + n \mu)}{n\mu (\lambda + \nu + n\mu)}.
\]

Furthermore, substituting these explicit expressions into (\ref{eq:r03}) and arranging the result, we obtain 
\[
	r^{(0,n)}_2 = \frac{ \lambda_2 (\lambda+n\mu) [ \lambda_1 (\lambda + \nu + (n+1)\mu) + \lambda_2 \nu  ]   }{ n \mu (\lambda + \nu + n \mu) (3 \lambda + 2 \nu + 2 (n+1)\mu ) \nu  }.
\]
Similarly, we also obtain 
\[
	r^{(1,n)}_2 = \frac{\lambda}{\nu} \left[ \frac{ \lambda + \nu + (n+1)\mu  }{3\lambda + 2\nu + 2(n+1)\mu}    + \frac{(\lambda+n\mu) [ \lambda (\lambda + (n+1)\mu + \lambda_1 \nu ]  }{ n\mu (\lambda + \nu + n\mu) (3\lambda + 2\nu + 2(n+1)\mu) }  \right].
\]
\end{kei}

\section{Taylor series expansion}
In this section, we derive Taylor series expansion for all non-zero elements of the rate matrices. 
In particular, we find Taylor series expansion of $r_{i}^{(0,n)}$ and $\ r_{i}^{(1,n)}\ (i=0,1,\dots,c)$ in terms of $1/n$. 
We use $\{ \theta_{m}^{(0,k)};m \in Z_{+}\}$ and $\{ \theta_{m}^{(1,k)};m \in Z_{+}\}$ as the coefficients of Taylor series expansion, 
where $k$ denotes the number of idle servers. We use the convention that if $k<0$ or $c<k$ then $\theta_{m}^{(0,k)}=0$ and $\theta_{m}^{(1,k)}=0$.   
Furthermore, $o(x)$ implies $\lim_{x \to 0} o(x) / x = 0$ and $O(x)$ implies $\limsup_{x \to 0} | O(x) / x | < \infty$, respectively.

In this section, Lemma~\ref{lem:3-1} gives the one term expansion while Lemma~\ref{lem:3-2} improves Lemma~\ref{lem:3-1} by replacing the small order $o(\cdot)$ by the big order $O(\cdot)$. 
Furthermore, Theorem~\ref{theo:5} provides the general expansion formulae for higher order Taylor series expansion of $r_{i}^{(0,n)}$ and $\ r_{i}^{(1,n)}\ (i=0,1,\dots,c)$. 

\begin{lem} \label{lem:3-1} {\rm
We have one term series expansion for the elements of $\vc{r}^{(0,n)},\ \vc{r}^{(1,n)}$ ($n  \to \infty $) as follows. 
\begin{eqnarray*}
 r_{c-k}^{(0,n)} &=& \theta_{0}^{(0,k)} \frac{1}{n^{k}} + o(\frac{1}{n^{k}}),\qquad k=0,1,\dots,c, 
\label{eq:3-1-1} \\
 r_{c-k}^{(1,n)} &=& \theta_{0}^{(1,k)} \frac{1}{n^{k}} + o(\frac{1}{n^{k}}), \qquad
 k=0,1,\dots,c, 
\label{eq:3-1-2} 
\end{eqnarray*}
where the sequences $\{ \theta_{0}^{(0,k)};k = 0,1,\dots,c \}$ and $\{ \theta_{0}^{(1,k)};k =  0,1,\dots,c \}$ are given as follows.
\begin{eqnarray*}
\theta_{0}^{(0,k)}=
\left\{ \begin{array}{ll}
0,&k=0, \\
 \displaystyle  \frac{\lambda_2 }{\mu},& k=1,\\
  \displaystyle   \frac{\lambda_2 }{\mu}
\prod_{i=1}^{k-1}\frac{(c-i)\nu}{\mu},  &k =2,\ \dots ,c. 
\end{array} \right.  
\end{eqnarray*}
\begin{eqnarray*}
\theta_{0}^{(1,k)}=
\left\{ \begin{array}{ll}
\vspace{2mm}
 \displaystyle \frac{\lambda}{c\nu},&k=0, \\
\displaystyle  \frac{\lambda}{\mu},&k=1,\\
  \displaystyle  \frac{\lambda}{\mu}
\prod_{i=1}^{k-1}\frac{(c-i)\nu}{\mu},\  &k =2,\ \dots ,c.  \\
\end{array} \right. 
\end{eqnarray*}
}
\end{lem}
\begin{proof}
The technical details are provided in Appendix~\ref{app:A}. 
\end{proof}

%
\begin{lem} \label{lem:3-2}
The expansion formulae in Lemma~\ref{lem:3-1} for $\vc{r}^{(0,n)}$ and $\vc{r}^{(1,n)}$ ($n  \to \infty $) can be improved as
\begin{eqnarray}
 r_{c-k}^{(0,n)} &=& \theta_{0}^{(0,k)} \frac{1}{n^{k}} + O(\frac{1}{n^{k+1}}),\qquad
 k=0,1,\dots,c, \qquad 
\label{eq:55} \\
 r_{c-k}^{(1,n)} &=& \theta_{0}^{(1,k)} \frac{1}{n^{k}} + O(\frac{1}{n^{k+1}}),\qquad
 k=0,1,\dots,c. \qquad
\label{eq:56} 
\end{eqnarray}
\end{lem}
\begin{proof}
The technical details are provided in Appendix~\ref{app:B}. 
\end{proof}
%
%
\begin{theo}\label{theo:5} 
The elements of $\vc{r}^{(0,n)}$ and $\vc{r}^{(1,n)}$ ($n \to \infty$) are given by  
\begin{eqnarray}
r_{c-k}^{(0,n)} = \sum_{i=0}^{m} \theta_{i}^{(0,k)} (-1)^i \frac{1}{n^{k+i}} + O(\frac{1}{n^{k+m+1}})
,\quad
 m \in \mathbb{N}, \label{eq:theo3.1}\\
r_{c-k}^{(1,n)} = \sum_{i=0}^{m} \theta_{i}^{(1,k)} (-1)^i \frac{1}{n^{k+i}} + O(\frac{1}{n^{k+m+1}})
,\quad
 m \in \mathbb{N}, \label{eq:theo3.2}
\end{eqnarray}
where $\{ \theta_{m}^{(0,k)} , \theta_{m}^{(1,k)};k = 0,1,\dots,c,\ m \in \mathbb{N} \}$ are recursively defined as follows.
\begin{eqnarray*}
\theta_{m}^{(0,0)} &=& - \frac{\lambda_1}{c\nu} \theta_{m-1}^{(0,1)} 
+ \frac{\mu}{c\nu} \sum_{j=0}^{m-1} \Phi_{j}^{(0,0)} \theta_{m-j-1}^{(0,1)} (-1)^{j+1} \nonumber \\
                        &  & \mbox{} + \frac{\mu}{c\nu} \sum_{j=1}^{m} \widetilde{\Phi}_{j}^{(1,0)} \theta_{m-j}^{(0,0)} (-1)^j ,\\
\theta_{m}^{(0,1)} &=& \sum_{j=2}^{\min(c,m+1)} \theta^{(0,j)}_{m+1-j}(-1)^{j}
,\\ 
\theta_m^{(0,k)} &=& \frac{(c-k+1)\nu}{\mu} \theta_m^{(0,k-1)} + \frac{\lambda}{\mu} \theta_{m-2}^{(0,k+1)}   \\ 
 && + \frac{\lambda+(c-k)\nu}{\mu} \theta_{m-1}^{(0,k)} + \sum_{j=0}^{m-2} \Phi_{j}^{(0,k)} \theta_{m-j-2}^{(0,1)} (-1)^{j} \\
 && + \sum_{j=0}^{m-1} \Phi_{j}^{(1,k)} \theta_{m-j-1}^{(0,0)} (-1)^{j+1}, \\ 
 & & k=2,3,\dots,c,\ \\
\end{eqnarray*}
\begin{eqnarray*}
\theta_{m}^{(1,0)} &=& - \frac{\lambda_1}{c\nu} \theta_{m-1}^{(1,1)} 
+ \frac{\mu}{c\nu} \sum_{j=0}^{m-1} \Phi_{j}^{(0,0)} \theta_{m-j-1}^{(1,1)} (-1)^{j+1} \nonumber \\
                        &   & \mbox{} + \frac{\mu}{c\nu} \sum_{j=1}^{m} \widetilde{\Phi}_{j}^{(1,0)} \theta_{m-j}^{(1,0)} (-1)^j ,\\
\theta_{m}^{(1,1)} &=& \sum_{j=2}^{\min(c,m+1)} \theta^{(1,j)}_{m+1-j}(-1)^{j},\\
\theta_m^{(1,k)} &=&  \frac{(c-k+1)\nu}{\mu} \theta_m^{(1,k-1)}
+ \frac{\lambda}{\mu} \theta_{m-2}^{(1,k+1)} 
 \\
 && \mbox{} + \frac{\lambda+(c-k)\nu}{\mu} \theta_{m-1}^{(1,k)} + \sum_{j=0}^{m-2} \Phi_{j}^{(0,k)} \theta_{m-j-2}^{(1,1)} (-1)^{j} \\
 & & \mbox{} + \sum_{j=0}^{m-1} \Phi_{j}^{(1,k)} \theta_{m-j-1}^{(1,0)} (-1)^{j+1},\\ 
 & & k=2,3,\dots,c.
\end{eqnarray*}
Furthermore, 
\begin{eqnarray*}
\Phi_{j}^{(0,k)} &=& \sum_{i=0}^{j} \theta_{i}^{(0,k+1)} (-1)^{j} \frac{(k+i)_{j-i}}{(j-i)!}, \\
\Phi_{j}^{(1,k)} & = &  \sum_{i=0}^{j} \theta_{i}^{(1,k+1)} (-1)^{j}  \frac{(k+i)_{j-i}}{(j-i)!} ,\\ 
\widetilde{\Phi}_{j}^{(1,0)} &=& \sum_{i=1}^{j} \theta_{i}^{(1,1)} (-1)^{j} \frac{(i)_{j-i}}{(j-i)!},
\end{eqnarray*}
where $(\phi)_n \ (-\infty < \phi < \infty, n \in \mathbb{Z}_{+})$ denotes the Pochhammer symbol defined by
\begin{eqnarray*}
(\phi)_n = \left \{
\begin{array}{ll}
1, & n =0,\\
\phi (\phi + 1) \dots (\phi + n -1), & n \in \mathbb{N}.
\end{array}
\right.
\end{eqnarray*}
\end{theo}
\begin{proof}
The technical details are provided in Appendix~\ref{app:C}. 
\end{proof}


\section{Asymptotic upper bound}
In this section, we present the asymptotic upper bound for the stationary distribution. To this end, we use Lemmas~\ref{lem:1} and \ref{lem:2}.   
\begin{lem} \label{lem:1}
For a square matrix 
$\vc{A}=\left( 
\begin{array}{ccc}
a_{1,1} & \cdots & a_{1,n} \\
\vdots & \ddots & \vdots \\
a_{n,1} & \cdots & a_{n,n} \\
\end{array} 
\right)$,\ 
and a vector $\vc{x}=(x_1,x_2\dots,x_n)$, we have
\begin{eqnarray*}
||\vc{xA}||_1 \leq ||\vc{x}||_{1}  ||\vc{A}||_{\infty},  
\end{eqnarray*}
where 
$
||\vc{x}||_{1} = \sum_{i=1}^{n} |x_i| ,\ 
 ||\vc{A}||_{\infty} = \max_{1\leq i \leq n} \sum_{j=1}^{n} |a_{ij}| .\ 
$
\end{lem}
\begin{proof}
\begin{eqnarray*}
||\vc{xA}||_1 &=& \sum_{j=1}^{n} |x_{1}  a_{1j}+ x_{2} a_{2j} + \dots +  x_{n} a_{nj}| \\
&\leq& \sum_{i=1}^{n} |x_{i}|  \sum_{j=1}^{n} |a_{ij}|   \\
&\leq& \left( \sum_{i=1}^{n} |x_{i}| \right) \left( \max_{1\leq i \leq n} \sum_{j=1}^{n} |a_{ij}| \right) \\
&=&  ||\vc{x}||_{1}||\vc{A}||_{\infty} .
\end{eqnarray*}
\end{proof}

\begin{lem}[Fact 5 in \cite{analying_retrial}]\label{lem:2}
For an integer $N \ (\geq 1$) and $\hat{a}>0,\ \hat{b}$ satisfying $\hat{b} \neq \hat{a}m -m^2\ (m=0,1,\dots)$, we have
\begin{eqnarray*}
\prod_{j=N}^{n} \left( 1+ \frac{\hat{a}}{j} + \frac{\hat{b}}{j^2} \right) =
O\left( n^{\widehat{a}}\right) 
,\qquad  n\to\infty.
\end{eqnarray*}
\end{lem}

\begin{theo}\label{theo:4-1}
We define $\vc{\pi}'_{n} = (\pi_{c-1,n}, \pi_{c,n})$ in order to obtain  
\begin{eqnarray*} 
||\vc{\pi}'_{n}||_1 =
O\left( n^{a} \times {\left(\frac{\lambda}{c\nu}\right)}^n\right),\qquad n\to\infty,
\end{eqnarray*}
where $a= (c^2\nu+\lambda)/c\mu$.
\end{theo}

\begin{proof}
The proof uses Lemmas~\ref{lem:1} and \ref{lem:2}. We define some new notations as follows.
\begin{eqnarray*}
\vc{R}^{(n)'} &=& \left( 
\begin{array}{cc}
r_{c-1}^{(0,n)} & r_{c}^{(0,n)}   \\
r_{c-1}^{(1,n)}  & r_{c}^{(1,n)} \\
\end{array} 
\right),
\end{eqnarray*}
where 
\begin{eqnarray*}
r_{c-1}^{(0,n)} & = & \theta_0^{(0,1)}\frac{1}{n} - \theta_1^{(0,1)}\frac{1}{n^2} + O(\frac{1}{n^3}), \\
r_{c}^{(0,n)}     & =  & \theta_0^{(0,0)} -\theta_1^{(0,0)}\frac{1}{n} + \theta_2^{(0,0)}\frac{1}{n^2} + O(\frac{1}{n^3}),  \\
r_{c-1}^{(1,n)} & =  & \theta_0^{(1,1)}\frac{1}{n} - \theta_1^{(1,1)}\frac{1}{n^2} + O(\frac{1}{n^3}), \\
r_{c}^{(1,n)}    & =  & \theta_0^{(1,0)} -\theta_1^{(1,0)}\frac{1}{n} + \theta_2^{(1,0)}\frac{1}{n^2} + O(\frac{1}{n^3}),
\end{eqnarray*}
and 
\[
	\vc{\pi}'_{n} = (\pi_{c-1,n}, \pi_{c,n}).
\]

It follows from $\vc{\pi}_n = \vc{\pi}_{n-1} \vc{R}^{(n)}$ that $\vc{\pi}_n^\prime = \vc{\pi}_{n-1}^\prime {\vc{R}^{(n)}}^\prime$. Thus, applying Lemma~\ref{lem:1} repeatedly, we obtain
\begin{eqnarray*}
||\vc{\pi}'_{n}||_1
&\leq &  ||\vc{\pi}'_{0}||_{1}  ||\vc{R}^{(1)'}||_{\infty} \dots ||\vc{R}^{(n-1)'}||_{\infty}  ||\vc{R}^{(n)'}||_{\infty}.\ 
\end{eqnarray*}

For sufficiently large $n$, $||\vc{R}^{(n)'}||_{\infty}$ is given by
\begin{eqnarray*}
||\vc{R}^{(n)'}||_{\infty}
&=& |{r}_{c-1}^{(1,n)}| + |{r}_{c}^{(1,n)}| \\
&=& \theta_0^{(1,0)} + (\theta_0^{(1,1)} - \theta_1^{(1,0)})  \frac{1}{n} \\
&  & \mbox{} + ( \theta_2^{(1,0)} - \theta_1^{(1,1)}) \frac{1}{n^2} + O(\frac{1}{n^3}) \\
&=& \theta_0^{(1,0)} \left(  1+ \frac{\theta_0^{(1,1)} - \theta_1^{(1,0)}}{\theta_0^{(1,0)} n}
+\frac{\theta_2^{(1,0)} - \theta_1^{(1,1)}}{\theta_0^{(1,0)} n^2} \right)  \\
&  & \mbox{} + O(\frac{1}{n^3})\\
&=& \frac{\lambda}{c\nu} \left( 1 + \frac{a}{n} + \frac{b}{n^2}\right) 
+ O(\frac{1}{n^3}),
\end{eqnarray*}
where 
\[
	a = \frac{\theta_0^{(1,1)} - \theta_1^{(1,0)}}{\theta_0^{(1,0)}}, \qquad b = \frac{\theta_2^{(1,0)} - \theta_1^{(1,1)}}{\theta_0^{(1,0)}}.
\]

Thus, for parameters that satisfy Lemma~\ref{lem:2}, we have 
\begin{eqnarray*}
 ||\vc{\pi}'_{0}||_{1} \prod_{i=1}^{n} ||\vc{R}^{(i)'}||_{\infty} 
=  O\left( n^{a} \times {\left(\frac{\lambda}{c\nu}\right)}^n\right),\qquad n\to\infty.  
\end{eqnarray*}
implying the desired result.
\end{proof}

\begin{kei}
We have 
\[
\pi_{i,n} = O \left(n^{a-c+i} \times \left(\frac{\lambda }{c\nu}\right)^n\right),\ i=0,1,\dots,c, \quad n\to \infty.
\]
\begin{proof}
From $\vc{\pi}_n = \vc{\pi}_{n-1} \vc{R}^{(n)}$, we have
\[
\pi_{i,n} = \pi_{c-1,n-1} r_i^{(0,n)} + \pi_{c,n-1} r_i^{(1,n)},\qquad i=0,1,\dots,c.
\]
It follows from Theorem~\ref{theo:5} that 
\[
r_i^{(0,n)}=O(\frac{1}{n^{c-i}}),\qquad r_i^{(1,n)}=O(\frac{1}{n^{c-i}}),\qquad n\to\infty. 
\]
Theorem~\ref{theo:4-1} yields
\[
||\vc{\pi}'_{n}||_1 =
O\left( n^{a} \times {\left(\frac{\lambda}{c\nu}\right)}^n\right),\qquad n\to\infty.
\]
Thus, 
\[
\pi_{i,n} = O \left(n^{a-c+i} \times \left(\frac{\lambda }{c\nu}\right)^n\right),\qquad n\to \infty.
\]
\end{proof}
\end{kei}

\begin{remark}
In~\cite{analying_retrial}, only the last row of the rate matrices is nonzero. This fact allows us to evaluate the 
tail probability using the product of a sequence of scalars. However, since the last two rows of the rate matrices are nonzero in our model,  
we need to deal with the product of a sequence of matrices.  Thus, in order to apply the technique given in~\cite{analying_retrial}, i.e., Lemma~\ref{lem:2}, we need 
to use Lemma~\ref{lem:1}.   
\end{remark}


\section{Numerical algorithm}

In this section, we propose a computational algorithm for the stationary distribution of our model 
extending that proposed by~Phung-Duc et al.~\cite{sutikeisan} for the fundamental M/M/$c$/$c$ retrial queues without guard channels.  
In Section~\ref{subsec:5.1}, we show some results which are the basis for the algorithm. 
Section~\ref{sec:5.2} presents algorithms for the rate matrices and the stationary distribution. 
Section~\ref{sec:N0} proposes a simple method for determining the truncation point used in an algorithm in Section~\ref{sec:5.2}. 
Section~\ref{sec:5.4} derives some performance measures such as the blocking probability for fresh calls and that for handover and retrial calls. 


\subsection{Efficient computation} \label{subsec:5.1}
Due to Lemma~\ref{prop:1}, we need to compute $k$ inverse matrices in order to obtain $\vc{R}_{k}^{(n)}$. 
It may take a long time when the number of servers is large. 
Thus, instead of computing the inverse matrices, we propose a new method exploiting the fact that only the last two rows are nonzero. 
The computational complexity of our new method is only $O(c)$. In particular, the computational complexity in all the theorems and lemmas below are $O(c)$. 

It should be noted that the computation of $\vc{R}^{(n)}$ and $\vc{R}_{k}^{(n)}$ is equivalent to that of their 
last two rows $\vc{r}^{(n)}$ and $\vc{r}_k^{(n)}$, i.e.,
%
\begin{eqnarray}
\vc{r}^{(n)} = 
\left( 
\begin{array}{c}
\vc{r}_{}^{(0,n)}\\
\vc{r}_{}^{(1,n)}\\
\end{array} 
\right)
,\qquad 
\vc{r}_k^{(n)}  
 = 
\left( 
\begin{array}{c}
\vc{r}_{k}^{(0,n)}\\
\vc{r}_{k}^{(1,n)}\\
\end{array} 
\right), \label{eq:5.1}
\end{eqnarray}
where $\vc{r}_{}^{(i,n)}$ and $\vc{r}_{k}^{(i,n)}$ ($i=0,1$) are vectors of $c+1$ elements.
 
\begin{defi} \label{defi:5.1}{\rm
We define the function $r_{n}$ as follows. Let
\vc{X}(\vc{x},\ \vc{y}) =
$\left( 
\begin{array}{c}
\vc{O}\\
\vc{x} \\
\vc{y} \\
\end{array} 
\right)$ and 
{\rm Lr (\vc{Y})}=$\left( 
\begin{array}{c}
\vc{y}_0 \\
\vc{y}_1 \\
\end{array} 
\right)$,\
($\vc{y}_0 ,\ \vc{y}_1$ are the second last and the last rows of \vc{Y})
and
\begin{eqnarray*}
r_n
\left( 
\begin{array}{c}
\vc{x} \\
\vc{y} \\
\end{array} 
\right) =
{\rm Lr}(R_n(\vc{X}(\vc{x},\ 
\vc{y} ))).\ 
\end{eqnarray*} 
where $\vc{x}$ and $\vc{y}$ are vectors with an appropriate dimension. 

}\end{defi}
It is easy to see that $\vc{r}^{(n)}$ and $\vc{r}^{(n)}_{k}$ satisfies the following equations. 
\begin{eqnarray*}
\vc{r}^{(n)} & = & r_n(\vc{r}^{(n+1)}), \\
\vc{r}_k^{(n)} & = & r_n( \vc{r}_{k-1}^{(n+1)} )= r_n \circ r_{n+1} \circ \dots \circ r_{n+k-1}
\left( \vc{O}
\right), 
\end{eqnarray*}
for $n,k \in \mathbb{N}$.
%
Lemmas~\ref{lem:5-1} and~~\ref{lem:5-2} compute $\vc{r}_{k}^{(0,n)}$ and $\vc{r}_{k}^{(1,n)}$ using $\vc{r}_{k-1}^{(n+1)}$, respectively. 
Furthermore, Lemma~\ref{lem:5-3} computes the stationary distribution of the censored Markov chain on level 0 using $\vc{r}^{(1)}$. 

\begin{lem}\label{lem:5-1}
{\rm
For arbitrary $n,k$, we have 
\begin{eqnarray*}
r_{k,i}^{(0,n)} = \alpha_{i} + \beta_{i} r_{k,c}^{(0,n)},\qquad i=0,1,\dots,c-1, 
\end{eqnarray*} 
where $\{ \alpha_{i} ,\beta_{i} ; i=0,1,\dots,c\}$ and $ r_{k,c}^{(0,n)}$ are given as follows.   
\begin{eqnarray*}
\alpha_c &=&0,\ \beta_c =1, \\ 
\alpha_{c-1} & = & 0, \quad \beta_{c-1}= - \frac{b_c^{(n)} + (n+1)\mu r_{k-1,c-1}^{(1,n+1)}}{\lambda_1 +(n+1)\mu r_{k-1,c-1}^{(0,n+1)}},\  \\
\alpha_{c-2} & = & - \frac{\lambda_2}{\lambda} ,\ \\
\beta_{c-2} &=& - \frac{b_{c-1}^{(n)}\beta_{c-1} +c\nu}{\lambda}  \\ 
                 & & \mbox{} -  \frac{(n+1)\mu r_{k-1,c-2}^{(0,n+1)}\beta_{c-1} +(n+1)\mu r_{k-1,c-2}^{(1,n+1)}}{\lambda}, \\
\alpha_{i-1} &=&- \frac{b_i^{(n)}\alpha_i +(i+1)\nu\alpha_{i+1}}{\lambda},  \\
                 &  & i=c-2,c-3,\dots,1, \\
\beta_{i-1}  &=&- \frac{b_i^{(n)}\beta_i + (i+1)\nu\beta_{i+1}}{\lambda} \\
                 &  & \mbox{} - \frac{(n+1)\mu r_{k-1,i-1}^{(0,n+1)} \beta_{c-1} +(n+1) \mu r_{k-1,i-1}^{(1,n+1)}}{\lambda}, \\
               &  & i = c-2, c-3,\dots,1,
\end{eqnarray*} 
and 
\[
r_{k,c}^{(0,n)} =- \frac{b_0^{(n)}\alpha_0+\nu\alpha_1}{b_0^{(n)}\beta_0+\nu\beta_1}.
\]}

\end{lem}
\begin{proof}
The technical details are provided in Appendix~\ref{app:D}
\end{proof}

\begin{lem} \label{lem:5-2}
For arbitrary $n$ and $k$, we have 
\begin{eqnarray*}
r_{k,i}^{(1,n)} = \alpha_{i} + \beta_{i} r_{k,c}^{(1,n)} ,\qquad i=0,1,\dots,c-1, 
\end{eqnarray*} 
where $\{ \alpha_{i} ,\beta_{i} ; i=0,1,\dots,c\}$ and $ r_{k,c}^{(1,n)}$ are given as follows. 
\begin{eqnarray*}
\alpha_c &=&0,\ \beta_c =1, \\
 \alpha_{c-1} & = &  - \frac{\lambda}{\lambda_1 +(n+1)\mu r_{k-1,c-1}^{(0,n+1)}}, \\
\beta_{c-1} &= & - \frac{b_c^{(n)} + (n+1)\mu r_{k-1,c-1}^{(1,n+1)}}{\lambda_1 +(n+1)\mu r_{k-1,c-1}^{(0,n+1)}}, \\
\alpha_{i-1} &=&- \frac{b_i^{(n)}\alpha_i +(i+1)\nu\alpha_{i+1}}{\lambda}  \\
                 &  & \mbox{} - \frac{(n+1)\mu r_{k-1,i-1}^{(0,n+1)}\alpha_{c-1}}{\lambda}, \\ 
                 &  & i=c-1,c-2,\dots,1, \\
\beta_{i-1} &=& - \frac{b_i^{(n)}\beta_i +(i+1)\nu\beta_{i+1}}{\lambda} \\
                &  & \mbox{} - \frac{(n+1)\mu r_{k-1,i-1}^{(0,n+1)}\beta_{c-1} +(n+1)\mu r_{k-1,i-1}^{(1,n+1)}}{\lambda}, \\ 
               &  & i=c-1,c-2,\dots,1.
\end{eqnarray*} 
Furthermore, 
\[
	r_{k,c}^{(1,n)} =  - \frac{b_0^{(n)}\alpha_0+\nu\alpha_1}{b_0^{(n)}\beta_0+\nu\beta_1}.
\]
\end{lem}
\begin{proof}
This lemma can be proved using the same technique as in Lemma~\ref{lem:5-1}.
\end{proof}

\begin{lem}\label{lem:5-3}
The solution $\vc{x}_0 = (x_0,\ x_1,\ \dots,x_c)$ for 
\begin{eqnarray*}
\vc{x}_{0} \left( \vc{Q}_{1}^{(0)}+\vc{R}^{(1)} \vc{Q}_{2}^{(1)} \right) &=& \vc{0},\qquad \vc{x}_{0}\vc{e} =1,\ 
\end{eqnarray*} 
is given by $x_{i}= \beta_{i} x_{c}\ (i=0,1,\dots,c)$, where 
$\{\beta_{i};i=0,1,\dots,c-1,c \}$ is recursively defined as   
\begin{eqnarray*}
\beta_c &=&1,\qquad 
\beta_{c-1}=\frac{\lambda+c\nu-\mu r_{c-1}^{(1,1)}}{\lambda_{1}+\mu r_{c-1}^{(0,1)}} ,\\
\beta_{i-1} &=&\frac{(\lambda+i\nu)\beta_{i}-(i+1)\nu \beta_{i+1}-\mu(r_{i-1}^{(0,1)}\beta_{c-1}+r_{i-1}^{(1,1)})}{\lambda}\\ 
                &  &  i=c-1,c-2,\dots,
\end{eqnarray*} 
and then 
\[
	{x}_{c} = \frac{1}{\beta_0 + \beta_1 + \beta_2 + \dots + \beta_c}. 
\]
\end{lem}

\begin{remark}
$\vc{x}_0$ is proportional to $\vc{\pi}_{0}$.
\end{remark} 

\begin{remark}
Computation of $\vc{r}_{k}^{(0,n)}$ and $\vc{r}_{k}^{(1,n)}$ using Lemmas \ref{lem:5-1} and \ref{lem:5-2} might be numerically unstable due to overflow. Thus, we use recursive formulae in Theorem~\ref{theo:5.3} to obtain a numerically stable scheme. 
\end{remark}

\begin{theo}\label{theo:5.3}{\rm
Sequence $\{x_i;i=0,1,\dots,c\}$ represents either $\{r_{i,k}^{(0,n)};i=0,1,\dots,c\}$ or $\{r_{i,k}^{(1,n)};i=0,1,\dots,c\}$.   
$\{ x_{i};i=0,1,\dots,c-2 \}$ is calculated in terms of $x_{c-1}$ and $x_c$ as follows. 
\begin{eqnarray*}
x_i = \frac{(i+1)\nu x_{i+1} + D_i}{B_i} ,\qquad i=0,1,\dots,c-2, 
\end{eqnarray*} 
where $\{ B_i,D_i ;i=0,1,\dots,c-2\}$ are given as follows.   
\begin{eqnarray*}
B_0&=&\lambda + n\mu,\qquad D_0 =0,\\
B_i&=&(\lambda +i\nu +n\mu) -\frac{\lambda i \nu}{B_{i-1}},\qquad
 i=1,2,\dots,c-2,\\  
D_i &=& (n+1)\mu(r_{k-1,i-1}^{(0,n+1)} x_{c-1} +r_{k-1,i-1}^{(1,n+1)} x_{c}) +\frac{\lambda D_{i-1}}{B_{i-1}} , \\ 
     &  & i=1,2,\dots,c-2.
\end{eqnarray*}}
In addition, 
\[
B_i >\lambda ,\qquad D_i>0. 
\]
\end{theo}

\begin{proof}
We prove using mathematical induction. Let $\{x_i;i=0,1,\dots,c\}$ denote $\{r_{i,k}^{(0,n)};i=0,1,\dots,c\}$ defined in (\ref{eq:5.1}). 
We have
\begin{eqnarray*}
 x_0 = \frac{\nu}{\lambda + n\mu} x_1. 
\end{eqnarray*}
Thus, $B_0=\lambda + n\mu$ and $D_0 =0$. For $i=1,2,\dots,c-2$, we prove by mathematical induction. 
For $j=1,2,\dots,i-1$, assuming that 
\begin{eqnarray*}
B_j&=&(\lambda +j\nu +n\mu) -\frac{\lambda j \nu}{B_{j-1}},\\
D_j &=& (n+1)\mu(r_{k-1,j-1}^{(0,n+1)} x_{c-1} +r_{k-1,j-1}^{(1,n+1)} x_{c}) +\frac{\lambda D_{j-1}}{B_{j-1}}, 
\end{eqnarray*}
are true, we show that it is also true for $j=i$. Using the assumption of mathematical induction, we have
\begin{eqnarray*}
\lambda  \frac{i \nu x_{i} + D_{i-1}}{B_{i-1}} - (\lambda +i \nu + n \mu) x_i +
(i+1) \nu x_{i+1} + \tilde{x}_{i} =0,
\end{eqnarray*}
where $\tilde{x}_{i} = (n+1)\mu(r_{k-1,i-1}^{(0,n+1)} x_{c-1} +r_{k-1,i-1}^{(1,n+1)} x_{c})$. 
Arranging this formula yields 
\begin{eqnarray*}
x_i &=& \frac{(i+1)\nu x_{i+1} + (\lambda D_{i-1}/B_{i-1} 
+ \tilde{x}_{i}) }{(\lambda + i\nu + n\mu) - \lambda i \nu / B_{i-1}} \\
&=& \frac{(i+1)\nu x_{i+1} + D_{i}}{B_{i}}, 
\end{eqnarray*}
implying that the case $j=i$ is also true. Thus, for any $i=1,2,\dots,c-2$, the desired result is established. 
We can show similar result for $\{r_{i,k}^{(1,n)};i=0,1,\dots,c\}$. 
\end{proof}

\begin{remark}
Using Theorem~\ref{theo:5.3}, we can calculate $r_{k,i}^{(0,n)}$ ($i=0,1,\dots,c-2$) in terms of $r_{k,c-1}^{(0,n)}$ and $r_{k,c}^{(0,n)}$, and $r_{k,i}^{(1,n)}$ ($i=0,1,\dots,c-2$) in terms of $r_{k,c-1}^{(1,n)}$ and $r_{k,c}^{(1,n)}$. 
Furthermore, $r_{k,c-1}^{(0,n)}$, $r_{k,c}^{(0,n)}$, $r_{k,c-1}^{(1,n)}$ and $r_{k,c}^{(1,n)}$ are obtained from Lemmas~\ref{lem:5-1} and \ref{lem:5-2}. 
\end{remark}


\subsection{Computational algorithm} \label{sec:5.2}
In this section, we present an algorithm for computing the rate matrices and then a procedure for the computation of the stationary distribution. 
Algorithm 1 shows a method for $\vc{r}^{(n)}$ while Algorithm 2 computes an approximation $\widehat{\vc{\pi}}
=(\widehat{\vc{\pi}}_0,\ \widehat{\vc{\pi}}_1,\ \dots,\ \widehat{\vc{\pi}}_{N})$ to the stationary distribution, where $\{ k_l ; l \in \mathbb{Z}_{+} \}$ is an arbitrary increasing sequence and $N$ is the truncation point given in advance. We will discuss how to choose the truncation point in Section~\ref{sec:N0}.
\begin{algorithm}[htbp]
\caption{ Computation of $\vc{r}^{(n)}$}\label{alg:1}
\begin{algorithmic}
\STATE {\bf Input:} $\{\vc{Q}_0^{(n)},\ \vc{Q}_1^{(n)},\ \vc{Q}_2^{(n)},\  k_n; n \in \mathbb{Z}_{+}\}
,\ \epsilon$
\STATE {\bf Output:} $\{ \widehat{\vc{r}}^{(n)} \}$
\STATE $l :=1;$
\STATE Compute $\vc{r}^{(n)}_{k_1}$ and $\vc{r}^{(n)}_{k_0}$ using Lemmas~\ref{lem:5-1}, \ref{lem:5-2} and Theorem~\ref{theo:5.3}. 
\WHILE{$||\vc{r}_{k_l}^{(n)} - \vc{r}_{k_{l-1}}^{(n)}||_{\infty} $ $>$ $\epsilon$}
\STATE $l := l+1;$
\STATE Compute $\vc{r}_{k_l}^{(n)}$ and $\vc{r}_{k_{l-1}}^{(n)}$ using Lemmas~\ref{lem:5-1}, \ref{lem:5-2} and Theorem~\ref{theo:5.3}. 
\STATE $\widehat{\vc{r}}^{(n)} := \vc{r}_{k_l}^{(n)}$;
\ENDWHILE
\end{algorithmic}
\end{algorithm}

\begin{algorithm}[htbp]
\caption{Stationary distribution}\label{alg:2}
\begin{algorithmic}
\STATE {\bf Input:} $  \lambda, \ \mu,\ \nu,\ c,\ \{k_n; n \in \mathbb{Z}_{+} \},\ \epsilon ,\ N$
\STATE {\bf Output:} $\{ \widehat{\vc{\pi}}_{n} ; n =0,1,\dots,N\}$
\STATE Compute $\widehat{\vc{r}}^{(N)}$ using Algorithm 1.\ 
\FOR{$n=1$ to $N-1$}
\STATE $\widehat{\vc{r}}^{(N-n)} := r_{N -n} (\widehat{\vc{r}}^{(N-n+1)})$;
\ENDFOR
\STATE Compute $\vc{x}_{0}$ using Lemma~\ref{lem:5-3}. 
\FOR{$n=1$ to $N$}
\STATE $\vc{x}_n := x_{c-1,n-1} \widehat{\vc{r}}^{(0,n)} + x_{c,n-1} \widehat{\vc{r}}^{(1,n)}$;
\ENDFOR
\FOR{$n=0$ to $N$}
\STATE $\widehat{\vc{\pi}}_{n} := \frac{\vc{x}_n}{\sum_{n=0}^{N} \vc{x}_n \vc{e}}$;
\ENDFOR 
\end{algorithmic}
\end{algorithm}

\subsection{Determination of the truncation point $N$} \label{sec:N0}
In Algorithm 2, the truncation point is given in advance and it should be large enough such that the tail probability is sufficiently small, i.e., 
\[
\sum_{n=N+1}^{\infty} \vc{\pi}_n \vc{e} < \epsilon,
\]
where $\epsilon$ is given in advance.

However, since $\vc{\pi}_n$ is not explicitly obtained for general M/M/$c$/$c$ retrial queues, a direct determination of such an $N$ is difficult. 
In this paper, we use the explicit results for an M/M/$1$/$1$ retrial queue to determine this truncation point. In particular, we consider an M/M/1/1 retrial queue with arrival rate $\lambda/c$, retrial rate $\mu$ and service rate $\nu$. 
This queue is stable since $\rho=\lambda/(c\nu)<1$ due to the stability condition of our original model. 

Let $p_{i,n}\ (i=0,1,n\in\mathbb{Z}_{+})$ denote the probability that the number of busy servers is $i$ and the number of customers in the orbit is $n$ in the 
M/M/1/1 retrial queue. It is shown in~\cite{sutikeisan} that
\begin{eqnarray*}
p_{0,n} & = & \frac{\rho^n}{n!}(1-\rho)^{\frac{\lambda}{c\mu}+1} \left( \frac{\lambda}{c\mu} \right)_{n}, \\ 
p_{1,n} & = & \frac{\rho^{n+1}}{n!}(1-\rho)^{\frac{\lambda}{c\mu}+1} \left(1+ \frac{\lambda}{c\mu} \right)_{n},
\end{eqnarray*} 
where $n \in \mathbb{Z}_{+}$ and $( \phi )_n \ ( - \infty < \phi < \infty,\ n \in \mathbb{Z}_{+})$ denotes the Pochhammer symbol defined by 
\begin{eqnarray*}
(\phi)_n = \left \{
\begin{array}{ll}
1, & n =0,\\
\phi (\phi + 1) \dots (\phi + n -1), & n \in \mathbb{N}.
\end{array}
\right.
\end{eqnarray*}

Using this result, we set the truncation point as follows.
\[
N = \inf \{n \mid \sum_{i=0}^{n} (p_{0,i} +p_{1,i}) > 1- \epsilon \},\qquad \epsilon >0.
\]
%
%
We verify the accuracy of this choice using numerical results.

\subsection{Blocking probability} \label{sec:5.4}
We derive blocking probabilities as performance measures. In our model, priority (handover) and retrial customers are blocked when all the servers are occupied while non-priority customers are blocked when $c-1$ servers are occupied.  
The blocking probability of low priority customers is given by 
\[	
\pi_{c-1} +\pi_c \ := \sum_{n=0}^{\infty}\pi_{c-1,n} +\sum_{n=0}^{\infty} \pi_{c,n},
\]
and the blocking probability of priority and retrial customers is given by 
\[
\pi_c := \sum_{n=0}^{\infty} \pi_{c,n}.
\]

\section{Numerical results}
In this section, we show some numerical examples. 
\subsection{Accuracy of Taylor series expansion}
The rate matrix is calculated using Algorithm 1 where 
the $\epsilon$ in Algorithm 1 is set to be sufficiently small and $k_n = 2^n$. 
Thus, we can say that the rate matrix obtained by Algorithm 1 is exact. 

First, we present some numerical examples to show the effectiveness of Taylor series expansion. 
Tables 1 and 2 show numerical results of $\vc{r}^{(n)}$ for $n=100$ and $n=1000$, respectively. 
Other parameters are given by $c=5,\mu=1,\ \nu=1,\ \lambda_2 / \lambda_1 =4$ and $\lambda$ is 
calculated from the traffic intensity $\rho  \ (= \lambda / c \nu)$. We obtain exact value for the 
rate matrices using the matrix continued fraction approach, i.e., Algorithm 1 with enough accuracy. 
The one, two and three term expansions ($m=1,2,3$) are expressed by $\vc{r}^{(n,1)}$, $\vc{r}^{(n,2)}$ and $\vc{r}^{(n,3)}$. 
In these tables, we show the relative errors, i.e., 
$||\vc{r}^{(n,1)} - \vc{r}^{(n)}||_{\infty}/ 
||\vc{r}^{(n)}||_{\infty}$,\ $||\vc{r}^{(n,2)} - \vc{r}^{(n)}||_{\infty}/ \||\vc{r}^{(n)}||_{\infty}$ and $||\vc{r}^{(n,3)} - \vc{r}^{(n)}||_{\infty}/ ||\vc{r}^{(n)}||_{\infty}$. We observe that Taylor series expansion gives a good approximation. 
The relative errors for the case $n=1000$ are smaller than those for the case $n=100$ which agrees with Taylor series expansion formulae. 
We also observe that the relative error increases with the traffic intensity. This suggests that we need more computational effort for the cases of relatively heavy load in comparison with those of relatively light load.

Figures~\ref{fig:r0} and~\ref{fig:r1} represent $r_c^{(0,n)}$ and $r_c^{(1,n)}$ against the number of expansion terms. The parameters are given by
$n=1000,\ c=100,\ \mu=1,\ \nu=1,\ \lambda_2 / \lambda_1 =24$ and $\rho=0.9$. 
We observe that Taylor series expansion converges to the exact value after about $5$ terms.

\begin{table}[htbp]
  \begin{center}
    \caption{Relative error for $\vc{r}^{(n)}$ ($n=100$)}
    \begin{tabular}{|c||c|c|c|} \hline
       ($\rho$) & One term & Two terms & Three terms \\ \hline \hline
 0.1&	0.0051053401&	0.0003425140&	0.0000228094\\ \hline
0.2& 0.0086100661&	0.0006446694&	0.0000491957\\ \hline
0.3&	0.0120849796&	0.0009702635&	0.0000821267\\ \hline
0.4&	0.0155304303&	0.0013188638&	0.0001219509\\ \hline
0.5&	0.0189467632&	0.0016900430&	0.0001690102\\ \hline
0.6&	0.0223343192&	0.0020833798&	0.0002236397\\ \hline
0.7&	0.0256934342&	0.0024984580&	0.0002861679\\ \hline
0.8&	0.0290244403&	0.0029348670&	0.0003569166\\ \hline
0.9&	0.0323276648&	0.0033922015&	0.0004362009\\ \hline
 \end{tabular}
 \end{center}
\label{table:1}
\end{table}

\begin{table}[htbp]
  \begin{center}
    \caption{Relative error for $\vc{r}^{(n)}$ ($n=1000$)}
    \begin{tabular}{|c||c|c|c|} \hline
    ($\rho$) & One term & Two terms & Three terms \\ \hline \hline
0.1&	0.0004109342&	0.0000030754&	0.0000000215\\ \hline
0.2&	0.0008055116&	0.0000063974&	0.0000000500\\ \hline
0.3&	0.0011997010&	0.0000100293&	0.0000000863\\ \hline
0.4&	0.0015935030&	0.0000139704&	0.0000001309\\ \hline
0.5&	0.0019869182&	0.0000182201&	0.0000001843\\ \hline
0.6&	0.0023799470&	0.0000227778&	0.0000002472\\ \hline
0.7&	0.0027725901&	0.0000276429&	0.0000003200\\ \hline
0.8&	0.0031648480&	0.0000328146&	0.0000004033\\ \hline
0.9&	0.0035567214&	0.0000382924&	0.0000004976\\ \hline
 \end{tabular}
 \end{center}
\label{table:2}
\end{table}

\begin{figure}[htbp]
	\begin{center} 
		\includegraphics[width=80mm]{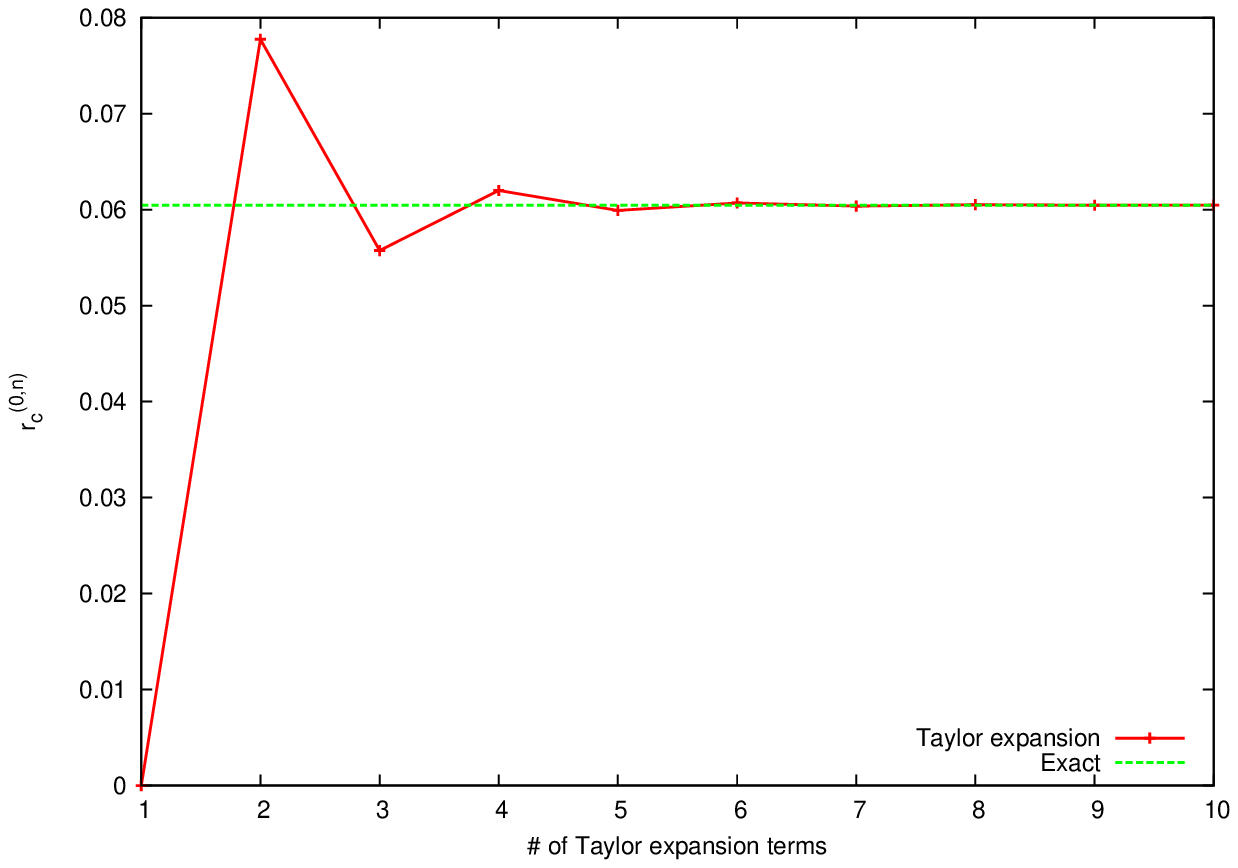}
    \end{center}
  \caption{$r_c^{(0,n)}$ vs. the \# of Taylor expansion terms.}
  \label{fig:r0}
\end{figure}

\begin{figure}[htbp]
   \begin{center}
		\includegraphics[width=80mm]{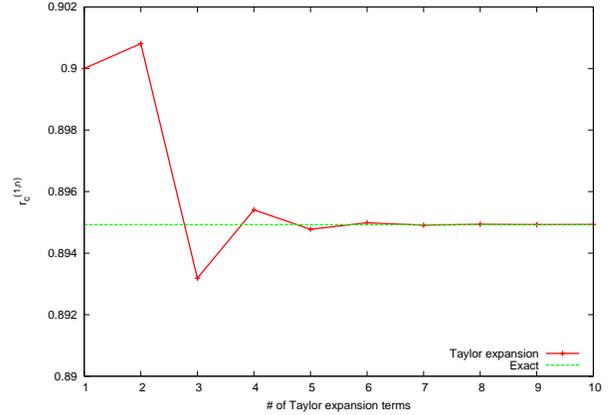}
  \end{center}
  \caption{$r_c^{(1,n)}$ vs. the \# of Taylor expansion terms.}
  \label{fig:r1}
\end{figure}

\subsection{Asymptotic behavior of $\pi_{i,n}/\rho^n$}
Figure~\ref{fig:pi/rho} shows $\pi_{i,n}/\rho^n \ (n \in \mathbb{Z}_+)$ against $n$ for some $i$. Parameters are given by
$c=100,\ N=1000,\ \mu=1,\ \nu=1/70, \lambda_1=1/25$ and $\lambda_2=24/25$. 
We observe that the five curves for $i=100, 75, 50, 25$ and 0 have negative slope. 
This implies that there should exist positive $C_1$, $C_2$ and $b$ such that
\[
C_1 \rho^n n^{-b} \leq \pi_{i,n} \leq C_2 \rho^n n^{-b},\qquad n\to \infty.
\]
Thus, the asymptotic results obtained in this paper can be further refined to be tighter. 
\begin{figure}[htbp]
 \begin{center}
  \includegraphics[width=80mm]{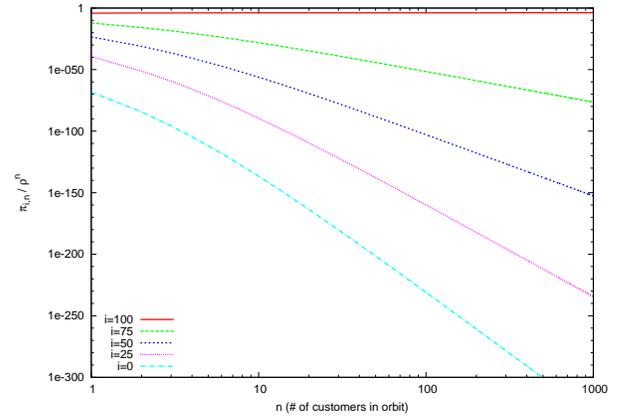}
 \end{center}
 \caption{$\pi_{i,n}/ \rho^n$ vs. the \# of customers in orbit\ ($n$).}
 \label{fig:pi/rho}
\end{figure}

\subsection{Blocking probability vs. number of servers}
We use the following parameters: $c = 100,\ \nu=1,\  \rho=0.7$ and $\lambda_2 / \lambda_1=24$. 
The truncation point $N$ is determined using the method in Section~\ref{sec:N0} for $\epsilon = 10^{-10}$. 
Blocking probabilities are $\pi_c$ and $\pi_{c-1} +\pi_c$  for high and low priority customers, respectively. Figure~\ref{fig:PB vs c} represents the blocking probabilities of two types of customers for three values of $\mu$ (0.1, 1 and 10). 
Obviously, for the same $\mu$, the the blocking probability for low priority customers is higher than that of high priority customers. 
Furthermore, the blocking probabilities increase with $\mu$ since customers who retry in a short interval may suffer from the same congested situation. 
An important observation is that all the curves are asymptotically linear when the number of servers is large. An asymptotic analysis for the case of large number of servers may be the topic of any future research.

\begin{figure}[htbp]
 \begin{center}
  \includegraphics[width=80mm]{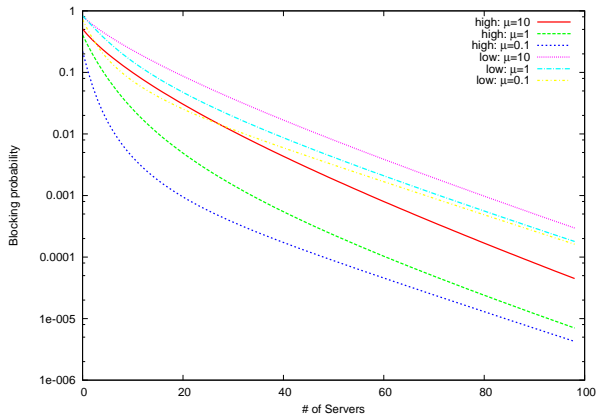}
 \end{center}
 \caption{Blocking probability vs. the \# of servers.}
 \label{fig:PB vs c}
\end{figure}

\subsection{Effect of the truncation point}
In this section, we investigate the effect of the truncation point. 
To this end, we define the absolute error $e_N$ for the number of busy servers as follows.
\[
e_N = \left| \frac{\lambda}{\nu} - \widehat{{\rm E}[C]} \right|,
\]
where $\widehat{{\rm E}[C]}$ is numerically calculated from our algorithms and $\lambda/\nu$ is its theoretical value due to Little law.

Figure~\ref{fig:aberror} shows the absolute error against the traffic intensity. 
Parameters are given by $c=25,50,100,200,\ \mu=1,\ \nu =1,\ 0.2 \leq \rho \leq 0.8$ and $\lambda_2 /\lambda_1=24$. 
Truncation point $N$ is determined using the method in Section~\ref{sec:N0} with $\epsilon = 10^{-10}$. We observe that the absolute error is small for any case. 


\begin{figure}[htbp]
 \begin{center}
  \includegraphics[width=80mm]{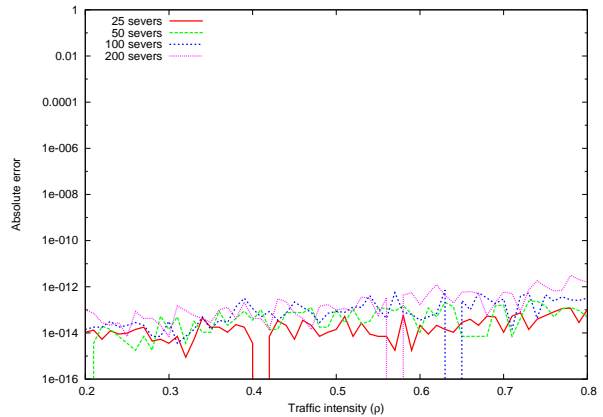}
 \end{center}
 \caption{Absolute error vs. Traffic intensity($\rho$).}
 \label{fig:aberror}
\end{figure}


\section{Concluding Remarks}
In this paper, we have introduced a new queueing model with guard channel for retrial and priority customers for cellular networks. 
The new queueing model is formulated using a QBD process which possesses a sparse structure allowing an efficient numerical algorithm 
and Taylor series expansion for all the nonzero elements of the rate matrices. We have also derived an asymptotic upper bound for the joint stationary 
distribution. Numerical results have revealed that the upper bound can be further improved. Future work includes finding the exact asymptotic formulae for the 
joint stationary distribution.

\appendix
\section{Proof of Lemma 2.1}\label{app:A0}

We prove Lemma 2.1 using Proposition~\ref{lem:twee}.
\begin{prop}\label{lem:twee}
{\rm (Tweedie \cite{tweedie} or Statement 8, p. 97 in~\cite{Retrial_queues})
Let $\{ \chi(t);t\geq 0 \}$ denote a Markov chain with the infinitesimal generator $\{ q_{s,p};s,p \in S\}$ on the state space $S$  $\sum_{p \in S} q_{s,p} = 0$. 
Furthermore, if the following conditions (i) and (ii) are satisfied, $\{\chi (t)\}$ is positive recurrent. 
\begin{description}
\item[(i)]  $\psi(s)$ ($s \in S$) is bounded from below. 
\item[(ii)] $y_s := \sum_{p \neq s} q_{sp} (\psi(p) - \psi(s))$. For any $s \in S$, $y_s < \infty$ and 
for any $s \in S$ except for a finite number of states, there exists a positive $\epsilon$ such that 
$y_s \leq - \epsilon$.
\end{description}
}
\end{prop}

\begin{proof}(Lemma 2.1)

\underline{$\bullet$ $\{X(t)\}$ is positive recurrent $\Rightarrow$ $\lambda/(c\nu)<1$}\\

Let $C$ denote the number of busy servers in the steady state. It follows from Little law that 
\[
\frac{\lambda}{\nu} = {\rm E} \left[ C \right].
\]
Thus, in order for ${X(t)}$ to be positive recurrent we must have $ {\rm E}  \left[ C \right] < c$ or equivalently $\lambda/(c\nu) < 1$.

\underline{$\bullet$ $\{ X(t) \}$ is positive recurrent $\Leftarrow$ $\lambda/(c\nu)<1$}\\

The transition rate of $\{ {X(t)} ; t \geq 0 \}$ is given by 
\[
q_{(i,j),(n,m)},\qquad (i,j),   (n,m) \in \mathcal{S}, 
\]
where $\mathcal{S} = \{ 0,1,\dots ,c \} \times  \mathbb{Z}_{+}$.

First, for $i=0,1,\dots,c-2$, 
\begin{eqnarray}
\lefteqn{q_{(i,j),(n,m)}=} \nonumber \\
& & \left\{ \begin{array}{ll}
\lambda, & (n,m)=(i+1,j), \\
i \nu, & (n,m)=(i-1,j), \\
j \mu,& (n,m)=(i+1,j-1), \\
-(\lambda +j \mu + i \nu), & (n,m)=(i,j), \\
0, & otherwise. \\
\end{array} \right.
\label{eq:A-1}
\end{eqnarray}
For $i=c-1$, 
\begin{eqnarray}
\lefteqn{q_{(c-1,j),(n,m)}=} \nonumber \\ 
& & \left\{ \begin{array}{ll}
\lambda_1, & (n,m)=(c,j), \\
(c-1) \nu, & (n,m)=(c-2,j), \\
\lambda_2, & (n,m)=(c-1,j+1), \\
j \mu,& (n,m)=(c,j-1), \\
-\left\{\lambda +j \mu + (c-1) \nu \right\}, & (n,m)=(c-1,j), \\
0, & otherwise. \\
\end{array} \right. \qquad
\label{eq:A-2}
\end{eqnarray}

For $i=c$, 
\begin{eqnarray}
q_{(c,j),(n,m)}=\left\{ \begin{array}{ll}
\lambda, & (n,m)=(c,j+1), \\
c \nu, & (n,m)=(c-1,j), \\
-(\lambda + c \nu), & (n,m)=(c,j), \\
0, & otherwise. \\
\end{array} \right.
\label{eq:A-3}
\end{eqnarray}

For $0<a<1$, we consider the test function $\phi(i,j)=ai +j$. 
For any $(i,j)$, we have $\phi(i,j) \geq 0$. 
Furthermore, $h(i,j)$ is defined as follows. 
\begin{eqnarray}
\lefteqn{h(i,j) =} \nonumber \\
& & \sum_{(n,m) \in \mathcal{S},\ (n,m) \neq (i,j)}  q_{(i,j),(n,m)} (\phi(n,m) - \phi(i,j)).\qquad
\end{eqnarray}


It follows from (\ref{eq:A-1}),\ (\ref{eq:A-2}) and (\ref{eq:A-3}) that
\begin{eqnarray*}
\lefteqn{h(i,j) =} \\ 
& & \left\{ \begin{array}{ll}
\lambda a - i\nu a +j \mu (a-1), &i=0,1,\dots,c-2,  \\
\lambda_1 a -(c-1) \nu a + j \mu (a-1) +\lambda_2, & i=c-1, \\
\lambda -c \nu a, & i=c. \\
\end{array} \right.
\end{eqnarray*}

Since $a<1$, for any $(i,j)$ we have $h(i,j) < \lambda$.
Furthermore, since $a<1$ for $i=0,1,\dots,c-1$ we have $\lim_{j \to \infty} h(i,j) = - \infty$. 
Thus, for any positive $\epsilon$, there exists $J(\epsilon)$ such that for $j > J(\epsilon)$ and $i=0,1,\dots,c-1$, we have $h(i,j) < - \epsilon$.

Next, we prove that except for a finite number of states, there exists $\epsilon > 0$ such that $h(i,j) < -\epsilon$. 
In order for $\lambda -c \nu a < 0$ except for a finite number of states, we choose $a$ such that
\[
\lambda -c \nu a < 0 \quad \Leftrightarrow \quad \rho =\lambda / (c\nu) < a<1.
\]

Thus, from the above formula and Lemma~\ref{lem:twee}, if $\lambda / (c\nu) < 1$ then $\{ X(t) \}$ is positive recurrent. 
\end{proof}

\section{Proof of Lemma 3.1} \label{app:A}


\begin{proof}
We prove that for $k = 0,1,\dots,c$, 
\begin{eqnarray*}
r_{c-k}^{(0,n)} &=& \theta_{0}^{(0,k)} \frac{1}{n^{k}} + o(\frac{1}{n^{k}}),\qquad
n \in \mathbb{N}, \\
r_{c-k}^{(1,n)} &=& \theta_{0}^{(1,k)} \frac{1}{n^{k}} + o(\frac{1}{n^{k}}),\qquad
n \in \mathbb{N},\\
r_{i}^{(0,n)} &=& o(\frac{1}{n^{k}}),\qquad i=0,1,\dots, c-k-1,\ \\
r_{i}^{(1,n)} &=& o(\frac{1}{n^{k}}),\qquad i=0,1,\dots, c-k-1 ,\ 
\end{eqnarray*}
by mathematical induction, where $i \in \emptyset$ if $k=c$. 

$\bullet$ {\bf The case $k=1$}

According to Lemma~\ref{lem:2-3}, for $i=0,1,2,\dots,c-1$
\begin{eqnarray}
r_i^{(0,n)} &=& o(1)  ,\qquad
r_i^{(1,n)} = o(1) \label{eq:12},\  \\
r_i^{(0,n)} &\leq& \frac{\lambda_2}{n\mu} ,\qquad
r_i^{(1,n)} \leq \frac{\lambda}{n\mu}. \label{eq:14}
\end{eqnarray}
Furthermore, it follows from (\ref{eq:r00}) and (\ref{eq:r10}) that
\begin{eqnarray}
r_0^{(0,n)} = \frac{1}{n\mu} \left( -\lambda r_0^{(0,n)} + \nu r_1^{(0,n)}\right), \label{eq:B3} \\
r_0^{(1,n)} = \frac{1}{n\mu} \left( -\lambda r_0^{(1,n)} + \nu r_1^{(1,n)}\right). \label{eq:B4}
\end{eqnarray}
From (\ref{eq:12}), (\ref{eq:B3}) and (\ref{eq:B4}), we obtain 
\begin{eqnarray}
r_0^{(0,n)} = o(\frac{1}{n}),\qquad r_0^{(1,n)} = o(\frac{1}{n}) \label{eq:578}.
\end{eqnarray}
In addition, it follows from (\ref{eq:r03}) and (\ref{eq:r12}) that
\begin{eqnarray*}
\lefteqn{(\lambda+c\nu) r_{c}^{(0,n)} = } \\ 
& &  \lambda_1 r_{c-1}^{(0,n)} +(n+1)\mu(r_{c-1}^{(0,n)} r_{c-1}^{(0,n+1)} + r_{c}^{(0,n)} r_{c-1}^{(1,n+1)}),\\
\lefteqn{(\lambda+c\nu) r_{c}^{(1,n)} = } \\
& &  \lambda_1 r_{c-1}^{(1,n)} +(n+1)\mu(r_{c-1}^{(1,n)} r_{c-1}^{(0,n+1)} + r_{c}^{(1,n)} r_{c-1}^{(1,n+1)}) +\lambda. 
\end{eqnarray*}
From (\ref{eq:14}), we obtain 
\begin{eqnarray*}
(\lambda+c\nu) r_{c}^{(0,n)} &\leq &  \lambda_1 r_{c-1}^{(0,n)} +\lambda_2 r_{c-1}^{(0,n)} + \lambda r_{c}^{(0,n)} ,\\
(\lambda+c\nu) r_{c}^{(1,n)} &\leq&   \lambda_1 r_{c-1}^{(1,n)} + \lambda_2  r_{c-1}^{(1,n)} +
\lambda r_{c}^{(1,n)}+\lambda. 
\end{eqnarray*}
Deleting $\lambda r_{c}^{(0,n)}$ and $\lambda r_{c}^{(1,n)}$ from both sides yields
\begin{eqnarray*}
c\nu r_{c}^{(0,n)} &\leq &  \lambda_1 r_{c-1}^{(0,n)} +\lambda_2 r_{c-1}^{(0,n)},\\
c\nu r_{c}^{(1,n)} &\leq&   \lambda_1 r_{c-1}^{(1,n)} + \lambda_2  r_{c-1}^{(1,n)} +\lambda. 
\end{eqnarray*}
 From (\ref{eq:12}), we obtain 
\begin{eqnarray}
r_c^{(0,n)} = o(1),\qquad \mbox{$r_c^{(1,n)} = O(1)$} \label{eq:347}  
\end{eqnarray}

From (\ref{eq:r01}) and (\ref{eq:14}), we have
\begin{eqnarray*}
r_{i}^{(0,n)} & =  &   \frac{\lambda r_{i-1}^{(0,n)} -(\lambda + i\nu) r_{i}^{(0,n)} +(i+1)\nu  r_{i+1}^{(0,n)} }{n\mu} \\ 
 & &        + \frac{(n+1)\mu  (r_{c-1}^{(0,n)}r_{i-1}^{(0,n+1)}+r_{c}^{(0,n)}r_{i-1}^{(1,n+1)})}{n\mu} \\
&\leq& \frac{\lambda r_{i-1}^{(0,n)} -(\lambda + i\nu) r_{i}^{(0,n)} +(i+1)\nu  r_{i+1}^{(0,n)}}{n\mu}  \\
& & \mbox{} + \frac{\lambda_2 r_{c-1}^{(0,n)}+\lambda r_c^{(0,n)}}{n\mu}.
\end{eqnarray*}
It follows from (\ref{eq:12}) and (\ref{eq:347}) that
\begin{eqnarray}
r_{i}^{(0,n)}= o(\frac{1}{n}),\qquad i=1,2,\dots,c-2. \label{eq:20}
\end{eqnarray}
From Lemma~\ref{lem:2-3}, (\ref{eq:578}) and (\ref{eq:20}), we obtain 
\begin{eqnarray}
r_{c-1}^{(0,n)}= \frac{\lambda_2 }{n\mu} +o(\frac{1}{n}).\   \label{eq:21}
\end{eqnarray}
Thus, we obtain $r_{i}^{(0,n)}= o(1/n),\ i=0,1,\dots,c-2,\  r_{c-1}^{(0,n)}= \theta_{0}^{(0,1)}/n + o(1/n)$.

Arranging (\ref{eq:r11}) yields
\begin{eqnarray}
\lefteqn{n\mu r_{i}^{(1,n)} = } \nonumber \\
& &  \lambda r_{i-1}^{(1,n)} -(\lambda + i\nu) r_{i}^{(1,n)} +(i+1)\nu  r_{i+1}^{(1,n)} 
+ \tilde{r}_i^{(1,n)}. \label{eq:r1nhe} \qquad 
\end{eqnarray}
It follows from (\ref{eq:12}), (\ref{eq:578}) and (\ref{eq:347}) that
\begin{eqnarray*}
r_{0}^{(1,n)}&=& o(\frac{1}{n}),\qquad r_{1}^{(1,n)}= o(1),\qquad r_{2}^{(1,n)}= o(1),\\
r_{0}^{(0,n+1)}&=& o(\frac{1}{n+1}),\qquad r_{0}^{(1,n+1)}= o(\frac{1}{n+1}),\\
r_{c-1}^{(1,n)}&=& o(1),\qquad r_{c}^{(1,n)}= O(1).
\end{eqnarray*}

Substituting the above formulae into (\ref{eq:r1nhe}) with $i=1$ yields $r_{1}^{(1,n)}= o(1/n)$. 
We assume that Lemma~\ref{lem:3-1} is true for $i=j-1$, i.e.,  $r_{j-1}^{(1,n)}= o(1/n)$.

From the preceding assumption, (\ref{eq:12}), (\ref{eq:578}) and (\ref{eq:347}), we have
\begin{eqnarray*}
r_{j-1}^{(1,n)}&=& o(\frac{1}{n}),\qquad r_{j}^{(1,n)}= o(1),\qquad r_{j+1}^{(1,n)}= o(1),\\
r_{j-1}^{(0,n+1)}&=& o(\frac{1}{n+1}),\qquad r_{j-1}^{(1,n+1)}= o(\frac{1}{n+1}),\\
r_{c-1}^{(1,n)}&=& o(1),\qquad r_{c}^{(1,n)}= O(1).
\end{eqnarray*}
Substituting these formulae into (\ref{eq:r1nhe}) with $i=j$, we obtain $r_{j}^{(1,n)}= o(1/n)$.

Using mathematical induction we have $r_{i}^{(1,n)}= o(1/n)$ for $i=1,2,\dots,c-2$, which together with Lemma~\ref{lem:2-3} and (\ref{eq:578}) yield
\begin{eqnarray}
r_{c-1}^{(1,n)}= \frac{\lambda}{n\mu} +o(\frac{1}{n}). \label{eq:k=1,c-1}
\end{eqnarray}
Thus, we obtain $  r_{i}^{(1,n)}= o(1/n),\ i=0,1,\dots,c-2,\  r_{c-1}^{(0,n)}= \theta_{0}^{(1,1)}/n + o(1/n)$.  

$\bullet$ {\bf The case $k=2,3,\dots, c-1$}\ 

It should be noted that the derivations for $r^{(0,n)}_{c-k}$ and $r^{(1,n)}_{c-k}$ are the same. Thus, we show for $r^{(0,n)}_{c-k}$ only. 
For $k=1,2,\dots,j$, we assume that
\begin{eqnarray}
 r_{c-k}^{(0,n)} &=& \theta_{0}^{(0,k)} \frac{1}{n^{k}} + o(\frac{1}{n^{k}}),\qquad n \in \mathbb{N},\label{eq:780}\\ 
 r_{c-k}^{(1,n)} &=& \theta_{0}^{(1,k)} \frac{1}{n^{k}} + o(\frac{1}{n^{k}}),\qquad n \in \mathbb{N},\
 \label{eq:781} \\
 r_{i}^{(0,n)} &=& o(\frac{1}{n^k}),\qquad i=0,1,\dots, c-k-1,\label{eq:782} \\
 r_{i}^{(1,n)} &=& o(\frac{1}{n^k}),\qquad i=0,1,\dots, c-k-1. \label{eq:783} 
\end{eqnarray}
We prove that the same expression is obtainable for the case $k=j+1$. Indeed, it follows from (\ref{eq:B3}),\ (\ref{eq:B4}),\ (\ref{eq:782}) and (\ref{eq:783}) that
\[
r_{0}^{(0,n)} =o(\frac{1}{n^{j+1}}),\qquad  r_{0}^{(1,n)} =o(\frac{1}{n^{j+1}}).
\]

For $i=1,2,\dots,c-j-2$, assuming that $ r_{i-1}^{(0,n)}= o(1/n^{j+1})$ and $r_{i-1}^{(1,n)}= o(1/n^{j+1})$,
we prove that $ r_{i}^{(0,n)}= o(1/n^{j+1})$ and  $r_{i}^{(1,n)}= o(1/n^{j+1})$. 

Indeed, arranging (\ref{eq:r01}) and (\ref{eq:r11}) yields
\begin{eqnarray}
\lefteqn{r_{i}^{(0,n)}  =} \nonumber \\ 
&& \frac{\lambda r_{i-1}^{(0,n)} -(\lambda + i\nu) r_{i}^{(0,n)} +(i+1)\nu  r_{i+1}^{(0,n)} +\tilde{r}_{i}^{(0,n)}}{n\mu}. \qquad \quad 
\label{eq:71}\\
\lefteqn{r_{i}^{(1,n)}  = } \nonumber \\
&&  \frac{\lambda r_{i-1}^{(1,n)} -(\lambda + i\nu) r_{i}^{(1,n)} +(i+1)\nu  r_{i+1}^{(1,n)} +\tilde{r}_{i}^{(1,n)}}{n\mu}. \qquad  \quad
\label{eq:72}
\end{eqnarray}
Applying the preceding assumption, (\ref{eq:347}), (\ref{eq:21}), (\ref{eq:782}) and (\ref{eq:783}) to (\ref{eq:71}) yields

\begin{eqnarray}
r_{i}^{(0,n)}= o(\frac{1}{n^{j+1}}),\qquad i=1,2,\dots,c-j-2.
\label{eq:3.16}
\end{eqnarray}
Similarly, substituting the preceding assumption, (\ref{eq:347}), (\ref{eq:k=1,c-1}), (\ref{eq:782}) and (\ref{eq:783}) to (\ref{eq:72}), we obtain  
\begin{eqnarray}
r_{i}^{(1,n)}= o(\frac{1}{n^{j+1}}),\qquad i=1,2,\dots,c-j-2.
\label{eq:3.17}
\end{eqnarray}

It follows from (\ref{eq:780}),\ (\ref{eq:781}),\ (\ref{eq:71}),\ (\ref{eq:72}),\ (\ref{eq:3.16}) and (\ref{eq:3.17}) that 
\begin{eqnarray*}
 r_{c-j-1}^{(0,n)} &=& \theta_{0}^{(0,j+1)}  \frac{1}{n^{j+1}} + o(\frac{1}{n^{j+1}}),\qquad n \in \mathbb{N}, \\
 r_{c-j-1}^{(1,n)} &=& \theta_{0}^{(1,j+1)}  \frac{1}{n^{j+1}} + o(\frac{1}{n^{j+1}}),\qquad n \in \mathbb{N}.
\end{eqnarray*}
Thus, we have proven the case $k=j+1$. As a result, we have proven for $k=2,3,\dots,c-1$. 


$\bullet$ {\bf The case $k=c$}

Substituting (\ref{eq:780}),\ (\ref{eq:781}),\ (\ref{eq:782}),\ (\ref{eq:783}) with $k=c-1$ into (\ref{eq:B3}) and (\ref{eq:B4}), we obtain 
\begin{eqnarray*}
r_{0}^{(0,n)} &=& \theta_{0}^{(0,c)} \frac{1}{n^{c}} + o(\frac{1}{n^{c}}),\qquad
n \in \mathbb{N}, \\
r_{0}^{(1,n)} &=& \theta_{0}^{(1,c)} \frac{1}{n^{c}} + o(\frac{1}{n^{c}}),\qquad
n \in \mathbb{N}.
\end{eqnarray*}

$\bullet$ {\bf The case $k=0$}

Arranging (\ref{eq:r03}) and (\ref{eq:r12}), we obtain 
\begin{eqnarray}
\lefteqn{(\lambda+c\nu) r_{c}^{(0,n)} = \lambda_1 r_{c-1}^{(0,n)}} \nonumber \\
& &  +(n+1)\mu(r_{c-1}^{(0,n)} r_{c-1}^{(0,n+1)} + r_{c}^{(0,n)} r_{c-1}^{(1,n+1)}), \label{eq:k=0,3}  \qquad  \\ 
\lefteqn{(\lambda+c\nu) r_{c}^{(1,n)} = \lambda_1 r_{c-1}^{(1,n)} } \nonumber \\
& & +(n+1)\mu(r_{c-1}^{(1,n)} r_{c-1}^{(0,n+1)} + r_{c}^{(1,n)} r_{c-1}^{(1,n+1)}) +\lambda. \label{eq:k=0,4} \qquad
\end{eqnarray}

From (\ref{eq:780}) and (\ref{eq:781}) with $k=1$, we obtain 
\begin{eqnarray*}
(n+1) r_{c-1}^{(0,n+1)} &=& \frac{\lambda_2}{\mu}  + o(1),\qquad n \in \mathbb{N},
\label{eq:k=0,1} \\
(n+1) r_{c-1}^{(1,n+1)} &=& \frac{\lambda}{\mu}  + o(1),\qquad n \in \mathbb{N}.
\label{eq:k=0,2}
\end{eqnarray*}

Substituting the above two formulae into (\ref{eq:k=0,3}) and (\ref{eq:k=0,4}) yields
\begin{eqnarray*}
(\lambda+c\nu) r_{c}^{(0,n)} &= &  \lambda_1 r_{c-1}^{(0,n)} + \lambda_2 r_{c-1}^{(0,n)}  +\lambda r_{c}^{(0,n)} + o(1),\qquad \\
(\lambda+c\nu) r_{c}^{(1,n)} &=&   \lambda_1 r_{c-1}^{(1,n)}+ \lambda_2 r_{c-1}^{(1,n)}  +\lambda r_{c}^{(1,n)} + o(1) +\lambda.\qquad
\end{eqnarray*}

Deleting $\lambda r_c^{(0,n)}$ and $\lambda r_c^{(1,n)}$ from both sides yields 
\begin{eqnarray*}
c\nu r_{c}^{(0,n)} &= &  \lambda_1 r_{c-1}^{(0,n)} + \lambda_2 r_{c-1}^{(0,n)}   + o(1), \\
c\nu r_{c}^{(1,n)} &=&   \lambda_1 r_{c-1}^{(1,n)}+ \lambda_2 r_{c-1}^{(1,n)}  + o(1) +\lambda.
\end{eqnarray*}
From these two formulae and the result for $k=1$, we obtain 
\begin{eqnarray*}
r_{c}^{(0,n)} &=& \theta_{0}^{(0,0)} + o(1),\qquad n \in \mathbb{N}, \\
r_{c}^{(1,n)} &=& \theta_{0}^{(1,0)} +  o(1),\qquad n \in \mathbb{N},
\end{eqnarray*}
where 
\begin{eqnarray*}
\theta_{0}^{(0,0)} = 0,\qquad \theta_{0}^{(1,0)} = \frac{\lambda}{c\nu}. 
\end{eqnarray*}
\end{proof}


\section{Proof of Lemma 3.2} \label{app:B}

\begin{proof}
We prove for $k=0,1,\dots,c$ using mathematical induction. 

$\bullet$ {\bf The case $k=1$}\\
From Lemma~\ref{lem:2-3}, we have 
\begin{eqnarray*}
r_{c-1}^{(0,n)} &=& \frac{ \lambda_2}{n\mu} - r_{c-2}^{(0,n)} - \sum_{k=3}^{c} r_{c-k}^{(0,n)}, \\
r_{c-1}^{(1,n)} &=& \frac{\lambda}{n\mu} -r_{c-2}^{(1,n)}- \sum_{k=3}^{c} r_{c-k}^{(1,n)}. 
\end{eqnarray*}
Furthermore, from Lemma~\ref{lem:3-1}, we have 
\begin{eqnarray*}
r_{c-2}^{(0,n)} = \theta_{0}^{(0,2)} \frac{1}{n^{2}} + o(\frac{1}{n^{2}}) = O(\frac{1}{n^{2}}), \\
r_{c-2}^{(1,n)} = \theta_{0}^{(1,2)} \frac{1}{n^{2}} + o(\frac{1}{n^{2}}) = O(\frac{1}{n^{2}}), \\
\sum_{k=3}^{c} r_{c-k}^{(0,n)} = O(\frac{1}{n^{2}}), \\
\sum_{k=3}^{c} r_{c-k}^{(1,n)} = O(\frac{1}{n^{2}}). 
\end{eqnarray*}
Thus, we obtain
\begin{eqnarray*}
 r_{c-1}^{(0,n)} &=& \theta_{0}^{(0,1)} \frac{1}{n} + O(\frac{1}{n^{2}}), \\
 r_{c-1}^{(1,n)} &=& \theta_{0}^{(1,1)} \frac{1}{n} + O(\frac{1}{n^{2}}). 
\end{eqnarray*}


$\bullet$ {\bf The case $k=2,3,\dots,c-1$}

We assume (\ref{eq:55}) and (\ref{eq:56}) are true for $r^{(0,n)}_{c-j}$ with $j = 1,2,\dots,k-1$, we prove that they are also true for $j = k$. 
Arranging (\ref{eq:r01}) and (\ref{eq:r11}) with $i=c-k$ yields
\begin{eqnarray}
\lefteqn{ r_{c-k}^{(0,n)} =  \frac{  \lambda r_{c-k-1}^{(0,n)}- \left\{\lambda + (c-k)\nu \right\} r_{c-k}^{(0,n)} }{n\mu} } \nonumber \\
& &\mbox{} + \frac{ (c-k+1) \nu  r_{c-k+1}^{(0,n)} +\tilde{r}_{c-k}^{(0,n)} }{n\mu},   \qquad  \label{eq:41}\\
\lefteqn{ r_{c-k}^{(1,n)} = \frac{ \lambda r_{c-k-1}^{(1,n)}-\{\lambda + (c-k)\nu \} r_{c-k}^{(1,n)}}{n\mu}}   \nonumber \\ 
& & \mbox{} + \frac{(c-k+1)\nu  r_{c-k+1}^{(1,n)} +\tilde{r}_{c-k}^{(1,n)}}{n\mu}.\qquad  \label{eq:42}
\end{eqnarray}

Applying the assumption of mathematical induction, Lemma~\ref{lem:3-1} and (\ref{eq:55}) with $k=1$, we obtain 
\begin{eqnarray*}
r_{c-k-1}^{(0,n)} &=&  \theta_{0}^{(0,k+1)} \frac{1}{n^{k+1}} + o(\frac{1}{n^{k+1}}), \\
r_{c-k}^{(0,n)} &=&  \theta_{0}^{(0,k)} \frac{1}{n^{k}} + o(\frac{1}{n^{k}}), \\
r_{c-k+1}^{(0,n)} &=&  \theta_{0}^{(0,k-1)} \frac{1}{n^{k-1}}  + O(\frac{1}{n^{k}}), \\
(n+1)r_{c-k-1}^{(0,n+1)} &=&  \theta_{0}^{(0,k+1)} \frac{1}{(n+1)^{k}} + o(\frac{1}{n^{k}}) \\ 
                                  & = & \theta_{0}^{(0,k+1)} \frac{1}{n^{k}} + o(\frac{1}{n^{k}}), \\
(n+1)r_{c-k-1}^{(1,n+1)} &=&  \theta_{0}^{(1,k+1)} \frac{1}{(n+1)^{k}} + o(\frac{1}{n^{k}}) \\ 
                                  & = & \theta_{0}^{(1,k+1)} \frac{1}{n^{k}} + o(\frac{1}{n^{k}}), \\
r_{c-1}^{(0,n)} &=& \theta_{0}^{(0,1)} \frac{1}{n} +  O( \frac{1}{n^2} ). \\
r_{c}^{(0,n)} &=&  0  + o(1),
\end{eqnarray*}
Thus, substituting the above formulae to (\ref{eq:41}) yields
\begin{eqnarray*}
 r_{c-k}^{(0,n)} &=&\frac{(c-k+1)\nu}{\mu}  \theta_{0}^{(0,k-1)} \frac{1}{n^{k}} + 
 O(\frac{1}{n^{k+1}}) \\
&=& \theta_{0}^{(0,k)} \frac{1}{n^{k}} + O(\frac{1}{n^{k+1}}). 
 \end{eqnarray*}

Similarly, it follows from the assumption of mathematical induction, Lemma~\ref{lem:3-1} and (\ref{eq:56}) with $k=1$ that
\begin{eqnarray*}
r_{c-k-1}^{(1,n)} &=&  \theta_{0}^{(1,k+1)} \frac{1}{n^{k+1}} + o(\frac{1}{n^{k+1}}), \\
r_{c-k}^{(1,n)} &=&  \theta_{0}^{(1,k)} \frac{1}{n^{k}} + o(\frac{1}{n^{k}}), \\
r_{c-k+1}^{(1,n)} &=&  \theta_{0}^{(1,k-1)} \frac{1}{n^{k-1}} + O(\frac{1}{n^{k}}), \\
(n+1)r_{c-k-1}^{(0,n+1)} &=&  \theta_{0}^{(1,k+1)} \frac{1}{(n+1)^{k}} + o(\frac{1}{n^{k}}) \\
                                 & = & \theta_{0}^{(0,k+1)} \frac{1}{n^{k}} + o(\frac{1}{n^{k}}), \\
(n+1)r_{c-k-1}^{(1,n+1)} &=&  \theta_{0}^{(1,k+1)} \frac{1}{(n+1)^{k}} + o(\frac{1}{n^{k}}) \\ 
                                 & = & \theta_{0}^{(1,k+1)} \frac{1}{n^{k}} + o(\frac{1}{n^{k}}), \\
r_{c-1}^{(1,n)} &=& \theta_{0}^{(1,1)} \frac{1}{n} +  O( \frac{1}{n^2} ), \\
r_{c}^{(1,n)} &=&  \theta_{0}^{(1,0)}  + o(1).
\end{eqnarray*}
Thus,  substituting these formulae into (\ref{eq:42}) yields
\begin{eqnarray*}
 r_{c-k}^{(1,n)} &=&\frac{(c-k+1)\nu}{\mu}  \theta_{0}^{(1,k-1)} \frac{1}{n^{k}}
+ O(\frac{1}{n^{k+1}}) \\
&=& \theta_{0}^{(1,k)} \frac{1}{n^{k}}+ O(\frac{1}{n^{k+1}}).
 \end{eqnarray*}
 Therefore, it follows from mathematical induction that (\ref{eq:55}) and (\ref{eq:56}) are true for $k=2,3,\dots,c-1$.

$\bullet$ {\bf The case $k=c$}

Lemma~\ref{lem:3-1} and (\ref{eq:55}) with $k=c-1$ and (\ref{eq:56}) yield
\begin{eqnarray*}
r_{1}^{(0,n)} &=& \theta_{0}^{(0,c-1)}\frac{1}{n^{c-1}} + O(\frac{1}{n^{c}}), \\
r_{1}^{(1,n)} &=& \theta_{0}^{(0,c-1)}\frac{1}{n^{c-1}} + O(\frac{1}{n^{c}}),\\
r_{0}^{(0,n)} &=&   \theta_{0}^{(0,c)}\frac{1}{n^{c}} + o(\frac{1}{n^{c}}) \\ 
                & = & O(\frac{1}{n^{c}}),\\
r_{0}^{(1,n)} &=&  \theta_{0}^{(1,c)}\frac{1}{n^{c}} + o(\frac{1}{n^{c}}) \\
                & = & O(\frac{1}{n^{c}}).
\end{eqnarray*}
Substituting the above formulae into (\ref{eq:B3}) and (\ref{eq:B4}), we obtain
 \begin{eqnarray*}
r_{0}^{(0,n)} &=& \frac{\nu}{\mu} \theta_{0}^{(0,c-1)} \frac{1}{n^c}  \\
                &  & \mbox{} + O(\frac{1}{n^{c+1}})
=\theta_{0}^{(0,c)} \frac{1}{n^c}  + O(\frac{1}{n^{c+1}}),\\
r_{0}^{(1,n)} &=& \frac{\nu}{\mu} \theta_{0}^{(1,c-1)} \frac{1}{n^c} + O(\frac{1}{n^{c+1}}) \\
                & = & \theta_{0}^{(1,c)} \frac{1}{n^c} + O(\frac{1}{n^{c+1}}).
\end{eqnarray*}

$\bullet$ {\bf The case $k=0$}

From (\ref{eq:55}) with $k=1$ and (\ref{eq:56}), we obtain 
\begin{eqnarray*}
(n+1) r_{c-1}^{(0,n)} &=& \theta_{0}^{(0,1)} + O(\frac{1}{n}),  \\
(n+1) r_{c-1}^{(1,n)} &=& \theta_{0}^{(1,1)} + O(\frac{1}{n}).
\end{eqnarray*}
Substituting the above two formulae into (\ref{eq:k=0,3}) and (\ref{eq:k=0,4}) yields, 
\begin{eqnarray*}
(\lambda+c\nu) r_{c}^{(0,n)} &= &  \lambda_1 r_{c-1}^{(0,n)} + \lambda_2 r_{c-1}^{(0,n)}  +\lambda r_{c}^{(0,n)} + O(\frac{1}{n}),\qquad \\
(\lambda+c\nu) r_{c}^{(1,n)} &=&   \lambda_1 r_{c-1}^{(1,n)}+ \lambda_2 r_{c-1}^{(1,n)}  +\lambda r_{c}^{(1,n)} +O(\frac{1}{n})+\lambda.\qquad
\end{eqnarray*}
Deleting $\lambda r_c^{(0,n)}$ and $\lambda r_c^{(1,n)}$ from both sides of the above formulae, we obtain 
\begin{eqnarray*}
c\nu r_{c}^{(0,n)} &= &  \lambda_1 r_{c-1}^{(0,n)} + \lambda_2 r_{c-1}^{(0,n)}   +  O(\frac{1}{n}), \\
c\nu r_{c}^{(1,n)} &=&   \lambda_1 r_{c-1}^{(1,n)}+ \lambda_2 r_{c-1}^{(1,n)}  +  O(\frac{1}{n}) +\lambda.
\end{eqnarray*}
From the result for $k=1$, we obtain
\begin{eqnarray*}
r_{c}^{(0,n)} &=& \theta_{0}^{(0,0)} + O(\frac{1}{n}),\qquad n \in \mathbb{N}, \\
r_{c}^{(1,n)} &=& \theta_{0}^{(1,0)} +  O(\frac{1}{n}),\qquad n \in \mathbb{N}.
\end{eqnarray*}
\end{proof}

\section{Proof of Theorem~3.1} \label{app:C}

\begin{proof}
We prove Theorem~\ref{theo:5} using mathematical induction. 
First, we show that Theorem~\ref{theo:5} is true for $m=1$. 

$\bullet$ {\bf The case $k=1$}

From Lemma~\ref{lem:2-3}, we have 
\begin{eqnarray*}
r_{c-1}^{(0,n)} &=&  \frac{ \lambda_2}{n\mu} - \sum_{i=0}^{c-2} r_{i}^{(0,n)}, \\
r_{c-1}^{(1,n)} &=&  \frac{ \lambda}{n\mu} - \sum_{i=0}^{c-2} r_{i}^{(1,n)}. 
\end{eqnarray*}

Lemma~\ref{lem:3-2} yields 
\begin{eqnarray*}
r_{c-2}^{(0,n)} &=& \theta_{0}^{(0,2)} \frac{1}{n^2}+ O( \frac{1}{n^3} ), \\
r_{c-2}^{(1,n)} &=& \theta_{0}^{(1,2)} \frac{1}{n^2}+ O( \frac{1}{n^3} ), \\
\sum_{i=0}^{c-3} r_{i}^{(0,n)} &=& O( \frac{1}{n^3} ), \\
\sum_{i=0}^{c-3} r_{i}^{(1,n)} &=& O( \frac{1}{n^3} ). 
\end{eqnarray*}
Thus, 
\begin{eqnarray*}
r_{c-1}^{(0,n)} &=&  \frac{\theta_{0}^{(0,1)}}{n} - \frac{\theta_{1}^{(0,1)}}{n^2} +  O( \frac{1}{n^3} ), \\
r_{c-1}^{(1,n)} &=&  \frac{\theta_{0}^{(1,1)}}{n} - \frac{\theta_{1}^{(1,1)}}{n^2} +  O( \frac{1}{n^3} ),
\end{eqnarray*}
where $ \theta_{1}^{(0,1)}=\theta_{0}^{(0,2)}$ and $\theta_{1}^{(1,1)}=\theta_{0}^{(1,2)}.$

$\bullet$ {\bf The case $k=2,3,\dots,c-1$}

Assuming that (\ref{eq:theo3.1}) and (\ref{eq:theo3.2}) in Theorem~\ref{theo:5} are true for 
$r^{(0,n)}_{c-j}$ and $r^{(1,n)}_{c-j}$ with $j = 1,2,\dots,k-1$, we prove that they are also true for $j = k$. 

Using the assumption of mathematical induction and Lemma~\ref{lem:3-2}, we obtain
\begin{eqnarray*}
r_{c-k-1}^{(0,n)} &=&  \theta_{0}^{(0,k+1)} \frac{1}{n^{k+1}} + O(\frac{1}{n^{k+2}}), \\
r_{c-k}^{(0,n)} &=&  \theta_{0}^{(0,k)} \frac{1}{n^{k}} + O(\frac{1}{n^{k+1}}), \\
r_{c-k+1}^{(0,n)} &=&  \theta_{0}^{(0,k-1)} \frac{1}{n^{k-1}} -\theta_{1}^{(0,k-1)} \frac{1}{n^{k}}  + O(\frac{1}{n^{k+1}}), \\
(n+1)r_{c-k-1}^{(0,n+1)} &=&  \theta_{0}^{(0,k+1)} \frac{1}{(n+1)^{k}} + O(\frac{1}{n^{k+1}})  \\ 
                                 & = & \theta_{0}^{(0,k+1)} \frac{1}{n^{k}} + O(\frac{1}{n^{k+1}}), \\
(n+1)r_{c-k-1}^{(1,n+1)} &=&  \theta_{0}^{(1,k+1)} \frac{1}{(n+1)^{k}} + O(\frac{1}{n^{k+1}}) \\ 
                                 & = & \theta_{0}^{(1,k+1)} \frac{1}{n^{k}} + O(\frac{1}{n^{k+1}}), \\
r_{c-1}^{(0,n)} &=& \theta_{0}^{(0,1)} \frac{1}{n} - \theta_{1}^{(0,1)} \frac{1}{n^2} +  O( \frac{1}{n^3} ), \\
r_{c}^{(0,n)} &=&  0  + O(\frac{1}{n}).
\end{eqnarray*}
Substituting these formulae into (\ref{eq:41}), we obtain
\begin{eqnarray*}
 r_{c-k}^{(0,n)} &=&\frac{(c-k+1)\nu}{\mu}  \theta_{0}^{(0,k-1)} \frac{1}{n^{k}}  \\ 
                    &  & \mbox{} - \{   \frac{\lambda + (c-k)\nu}{\mu}  \theta_{0}^{(0,k)} + \frac{(c-k+1)\nu}{\mu}  \theta_{1}^{(0,k-1)}   \} \frac{1}{n^{k+1}}  \\ 
                    &  & \mbox{} + O(\frac{1}{n^{k+2}}) \\
&=& \theta_{0}^{(0,k)} \frac{1}{n^{k}} - \theta_{1}^{(0,k)}  \frac{1}{n^{k+1}} + O(\frac{1}{n^{k+2}}), 
 \end{eqnarray*}
where 
\begin{eqnarray*}
\theta_{1}^{(0,k)} =    \frac{\lambda + (c-k)\nu}{\mu}  \theta_{0}^{(0,k)} +
 \frac{(c-k+1)\nu}{\mu}  \theta_{1}^{(0,k-1)}.   
 \end{eqnarray*}

Similarly, using the same methodology, we obtain
\begin{eqnarray*}
 r_{c-k}^{(1,n)}
= \theta_{0}^{(1,k)} \frac{1}{n^{k}} - \theta_{1}^{(1,k)}  \frac{1}{n^{k+1}} + O(\frac{1}{n^{k+2}}),
 \end{eqnarray*}
where
\begin{eqnarray*}
\theta_{1}^{(1,k)}  & = &    \frac{\lambda + (c-k)\nu}{\mu}  \theta_{0}^{(1,k)} +  \frac{(c-k+1)\nu}{\mu}  \theta_{1}^{(1,k-1)} \\ 
                         &   & \mbox{} - \theta_{0}^{(1,0)} \theta_{0}^{(1,k+1)}.  
 \end{eqnarray*}

$\bullet$ {\bf The case $k=c$}

Equations (\ref{eq:theo3.1}) and (\ref{eq:theo3.2}) with $k=c-1$ and Lemma~\ref{lem:3-2} yield
\begin{eqnarray*}
r_{1}^{(0,n)} &=& \theta_{0}^{(0,c-1)}\frac{1}{n^{c-1}} 
-\theta_{1}^{(0,c-1)}\frac{1}{n^{c}}
+ O(\frac{1}{n^{c+1}}), \\
r_{1}^{(1,n)} &=& \theta_{0}^{(0,c-1)}\frac{1}{n^{c-1}} 
-\theta_{1}^{(1,c-1)}\frac{1}{n^{c}}
+ O(\frac{1}{n^{c+1}}),\\
r_{0}^{(0,n)} &=&   \theta_{0}^{(0,c)}\frac{1}{n^{c}} + o(\frac{1}{n^{c}})
=O(\frac{1}{n^{c}}),\\
r_{0}^{(1,n)} &=&  \theta_{0}^{(1,c)}\frac{1}{n^{c}} + o(\frac{1}{n^{c}})
=O(\frac{1}{n^{c}}).
\end{eqnarray*}
Thus, (\ref{eq:B3}) and (\ref{eq:B4}) are written as follows.
\begin{eqnarray*}
r_{0}^{(0,n)} &=& \frac{\nu}{\mu} \theta_{0}^{(0,c-1)} \frac{1}{n^c}
- \left(  \frac{\lambda \theta_{0}^{(0,c)} + \nu \theta_{1}^{(0,c-1)}}{\mu} \right) \frac{1}{n^{c+1}}  \\ 
                  && \mbox{} + O(\frac{1}{n^{c+2}}),  \\
&=& \theta_{0}^{(0,c)} \frac{1}{n^c}  -\theta_{1}^{(0,c)} \frac{1}{n^{c+1}}    + O(\frac{1}{n^{c+2}}),\\
r_{0}^{(1,n)} &=& \frac{\nu}{\mu} \theta_{0}^{(1,c-1)} \frac{1}{n^c} 
- \left(  \frac{\lambda \theta_{0}^{(1,c)} + \nu \theta_{1}^{(1,c-1)}}{\mu} \right)  \frac{1}{n^{c+1}} \\ 
               &  & \mbox{} + O(\frac{1}{n^{c+2}}) \\
&=& \theta_{0}^{(1,c)} \frac{1}{n^c} -\theta_{1}^{(1,c)} \frac{1}{n^{c+1}}  + O(\frac{1}{n^{c+2}}),
\end{eqnarray*}
where
\begin{eqnarray*}
\theta_{1}^{(0,c)} = \frac{\lambda \theta_{0}^{(0,c)} + \nu \theta_{1}^{(0,c-1)}}{\mu},\qquad
\theta_{1}^{(1,c)} = \frac{\lambda \theta_{0}^{(1,c)} + \nu \theta_{1}^{(1,c-1)}}{\mu}.
\end{eqnarray*}

$\bullet$ {\bf The case $k=0$}

We use the same methodology as in Lemma~\ref{lem:3-2}. 
Equations (\ref{eq:theo3.1}) and (\ref{eq:theo3.2}) with $k=1$ and Lemma~\ref{lem:3-2} yield
\begin{eqnarray*}
r_{c-1}^{(0,n)} &=& \theta_{0}^{(0,1)} \frac{1}{n} -  \theta_{1}^{(0,1)} \frac{1}{n^2} + O(\frac{1}{n^3}),  \\
r_{c-1}^{(1,n)} &=& \theta_{0}^{(1,1)} \frac{1}{n} -  \theta_{1}^{(1,1)} \frac{1}{n^2} + O(\frac{1}{n^3}),  \\
 r_{c}^{(0,n)} &=& 0 +O(\frac{1}{n}),  \\
r_{c}^{(1,n)} &=& \frac{\lambda}{c\nu} +O(\frac{1}{n}).
\end{eqnarray*}
Thus, (\ref{eq:r03}) and (\ref{eq:r12}) are written as follows.
\begin{eqnarray*}
r_{c}^{(0,n)} &=& 0  + \frac{\lambda_1 \theta_0^{(0,1)} + \mu \theta_0^{(0,1)} \theta_0^{(0,1)}}{c\nu} \frac{1}{n} +O(\frac{1}{n^2})  \\
&=& \theta_{0}^{(0,0)} -\theta_{1}^{(0,0)}  \frac{1}{n} +O(\frac{1}{n^2}), \\ 
r_{c}^{(1,n)} &=& \frac{\lambda}{c\nu} 
+ \frac{\lambda_1 \theta_0^{(1,1)} + \mu \theta_0^{(1,1)}\theta_0^{(0,1)} 
- \mu \theta_0^{(1,0)} \theta_1^{(1,1)}}{c\nu} \frac{1}{n}
+O(\frac{1}{n^2}) \\
&=& \theta_{0}^{(1,0)}-\theta_{1}^{(1,0)}  \frac{1}{n} +O(\frac{1}{n^2}).
\end{eqnarray*}
Therefore, Theorem~\ref{theo:5} is established for $m=1$ and $k=0,1,\dots,c$.

Next, assuming that Theorem~\ref{theo:5} is true for $m-1$ ($m$ terms expansion), we prove that it is also true for $m$ ($m+1$ terms expansion). 

$\bullet$ {\bf The case $k=1$}\\
Lemma~\ref{lem:2-3} and mathematical induction yield
\begin{eqnarray*}
r_{c-1}^{(0,n)} &=&   \frac{\lambda_2 }{n\mu}  - \sum_{k=2}^{c} r_{c-k}^{(0,n)} \\
&=& \frac{\lambda_2 }{n\mu}   - \sum_{k=2}^{c}\left\{ \sum_{j=0}^{m-1} \theta_{j}^{(0,k)} (-1)^j \frac{1}{n^{k+j}} + O(\frac{1}{n^{k+m}}) \right\} \\
&=& \frac{\lambda_2 }{n\mu} -  \sum_{k=2}^{c}  \sum_{j=0}^{m-1} \theta_{j}^{(0,k)} (-1)^j \frac{1}{n^{k+j}} + O(\frac{1}{n^{m+2}}) \\
&=& \sum_{i=0}^{m} \theta_{i}^{(0,1)} (-1)^i \frac{1}{n^{1+i}} + O(\frac{1}{n^{m+2}}),
\end{eqnarray*}
where
\begin{eqnarray*}
\theta_{i}^{(0,1)} =\sum_{j=2}^{\min(c,i+1)} \theta_{i+1-j}^{(0,j)} (-1)^j.
\end{eqnarray*}

Similarly, we have

\begin{eqnarray*}
r_{c-1}^{(1,n)} &=&   \frac{\lambda}{n\mu}  - \sum_{k=2}^{c} r_{c-k}^{(1,n)} \\
&=& \frac{\lambda}{n\mu}   - \sum_{k=2}^{c}\left\{ \sum_{j=0}^{m-1} \theta_{j}^{(1,k)} (-1)^j \frac{1}{n^{k+j}} + O(\frac{1}{n^{k+m}}) \right\} \\
&=& \frac{\lambda}{n\mu} -  \sum_{k=2}^{c}  \sum_{j=0}^{m-1} \theta_{j}^{(1,k)} (-1)^j \frac{1}{n^{k+j}} + O(\frac{1}{n^{m+2}}) \\
&=& \sum_{i=0}^{m} \theta_{i}^{(1,1)} (-1)^i \frac{1}{n^{1+i}} + O(\frac{1}{n^{m+2}}),
\end{eqnarray*}
where
\begin{eqnarray*}
\theta_{i}^{(1,1)} =\sum_{j=2}^{\min(c,i+1)} \theta_{i+1-j}^{(1,j)} (-1)^j.
\end{eqnarray*}

$\bullet$ {\bf The case $k=2,3,\dots,c-1$}\\

Assuming that (\ref{eq:theo3.1}) and (\ref{eq:theo3.2}) in Theorem~\ref{theo:5} are true for 
$r^{(0,n)}_{c-j}$ and $r^{(1,n)}_{c-j}$ with $j = 1,2,\dots,k-1$, we prove that they are also true for $j = k$. 
Applying the assumption of mathematical induction and (\ref{eq:theo3.1}) for $k=1$ yields
\begin{eqnarray*}
r_{c-k-1}^{(0,n)} &=& \sum_{i=0}^{m-1} \theta_{i}^{(0,k+1)} (-1)^i \frac{1}{n^{k+i+1}} + O(\frac{1}{n^{k+m+1}}), \\
r_{c-k}^{(0,n)} &=& \sum_{i=0}^{m-1} \theta_{i}^{(0,k)} (-1)^i \frac{1}{n^{k+i}} + O(\frac{1}{n^{k+m}}), \\
r_{c-k+1}^{(0,n)} &=& \sum_{i=0}^{m} \theta_{i}^{(0,k-1)} (-1)^i \frac{1}{n^{k+i-1}} + O(\frac{1}{n^{k+m}}), \\
r_{c-1}^{(0,n)} &=& \sum_{i=0}^{m} \theta_{i}^{(0,1)} (-1)^i \frac{1}{n^{k+i}} + O(\frac{1}{n^{m+2}}), \\
r_{c}^{(0,n)} &=& \sum_{i=0}^{m-1} \theta_{i}^{(0,0)} (-1)^i \frac{1}{n^{k+i}} + O(\frac{1}{n^{m}}), \\
(n+1)r_{c-k-1}^{(0,n)} &=& \sum_{i=0}^{m-1} \theta_{i}^{(0,k+1)} (-1)^i \frac{1}{(n+1)^{k+i}} + O(\frac{1}{n^{k+m}})\\
&=& \sum_{j=0}^{m-1} \Phi_{j}^{(0,k)} \frac{1}{n^{k+j}} + O(\frac{1}{n^{k+m}}), \\
(n+1)r_{c-k-1}^{(1,n)} &=& \sum_{i=0}^{m-1} \theta_{i}^{(1,k+1)} (-1)^i \frac{1}{(n+1)^{k+i}} + O(\frac{1}{n^{k+m}})\\
&=& \sum_{j=0}^{m-1} \Phi_{j}^{(1,k)} \frac{1}{n^{k+j}} + O(\frac{1}{n^{k+m}}). \\
\end{eqnarray*}

Substituting these formulae into (\ref{eq:41}) and attracting the coefficient of $1/n^{k+m}$ of (\ref{eq:41}) and arranging the result, we obtain
\begin{eqnarray*}
 \theta_{m}^{(0,k)} := 
\frac{\lambda}{\mu} \theta_{m-2}^{(0,k+1)} +\frac{\lambda +(c-k)\nu}{\mu} \theta_{m-1}^{(0,k)}
+\frac{(c-k+1)\nu}{\mu}\theta_{m}^{(0,k-1)} \\
+\sum_{j=0}^{m-1} \Phi_{j}^{(0,k)} \theta_{m-j-2}^{(0,1)} (-1)^j 
+\sum_{j=0}^{m-1} \Phi_{j}^{(1,k)} \theta_{m-j-1}^{(0,0)} (-1)^{j+1}.
\end{eqnarray*}

Similarly, using the assumption of mathematical induction and (\ref{eq:theo3.2}) with $k=1$, we obtain
\begin{eqnarray*}
r_{c-k-1}^{(1,n)} &=& \sum_{i=0}^{m-1} \theta_{i}^{(1,k+1)} (-1)^i \frac{1}{n^{k+i+1}} + O(\frac{1}{n^{k+m+1}}), \\
r_{c-k}^{(1,n)} &=& \sum_{i=0}^{m-1} \theta_{i}^{(1,k)} (-1)^i \frac{1}{n^{k+i}} + O(\frac{1}{n^{k+m}}) ,\\
r_{c-k+1}^{(1,n)} &=& \sum_{i=0}^{m} \theta_{i}^{(1,k-1)} (-1)^i \frac{1}{n^{k+i-1}} + O(\frac{1}{n^{k+m}}), \\
r_{c-1}^{(1,n)} &=& \sum_{i=0}^{m} \theta_{i}^{(1,1)} (-1)^i \frac{1}{n^{k+i}} + O(\frac{1}{n^{m+2}}), \\
r_{c}^{(1,n)} &=& \sum_{i=0}^{m-1} \theta_{i}^{(1,0)} (-1)^i \frac{1}{n^{k+i}} + O(\frac{1}{n^{m}}), \\
(n+1)r_{c-k-1}^{(0,n)} &=& \sum_{i=0}^{m-1} \theta_{i}^{(0,k+1)} (-1)^i \frac{1}{(n+1)^{k+i}} + O(\frac{1}{n^{k+m}}), \\
&=& \sum_{j=0}^{m-1} \Phi_{j}^{(0,k)} \frac{1}{n^{k+j}} + O(\frac{1}{n^{k+m}}), \\
(n+1)r_{c-k-1}^{(1,n)} &=& \sum_{i=0}^{m-1} \theta_{i}^{(1,k+1)} (-1)^i \frac{1}{(n+1)^{k+i}} + O(\frac{1}{n^{k+m}}), \\
&=& \sum_{j=0}^{m-1} \Phi_{j}^{(1,k)} \frac{1}{n^{k+j}} + O(\frac{1}{n^{k+m}}). \\
\end{eqnarray*}

Substituting these formulae into (\ref{eq:42}) and extracting the coefficient of $1/n^{k+m}$ in (\ref{eq:42}) and arranging the result yields, 
\begin{eqnarray*}
\theta_{m}^{(1,k)} :=\frac{\lambda}{\mu} \theta_{m-2}^{(1,k+1)} +\frac{\lambda +(c-k)\nu}{\mu} \theta_{m-1}^{(1,k)}
+\frac{(c-k+1)\nu}{\mu}\theta_{m}^{(1,k-1)} \\
+\sum_{j=0}^{m-1} \Phi_{j}^{(0,k)} \theta_{m-j-2}^{(1,1)} (-1)^j 
+\sum_{j=0}^{m-1} \Phi_{j}^{(1,k)} \theta_{m-j-1}^{(1,0)} (-1)^{j+1}.
\end{eqnarray*}
Thus, we obtain the result for the case $k=2,3,\dots,c-1$. 

$\bullet$ {\bf The case $k=c$}\\

Using Lemma~\ref{lem:3-2}, (\ref{eq:theo3.1}) and (\ref{eq:theo3.2}) with $k=c-1$, we obtain 
\begin{eqnarray*}
r_{1}^{(0,n)} &=&  \sum_{i=0}^{m} \theta_{i}^{(0,c-1)} (-1)^i \frac{1}{n^{c+i-1}} + O(\frac{1}{n^{c+m}}), \\
r_{1}^{(1,n)} &=& \sum_{i=0}^{m} \theta_{i}^{(1,c-1)} (-1)^i \frac{1}{n^{c+i-1}} + O(\frac{1}{n^{c+m}}),\\
r_{0}^{(0,n)} &=&    \sum_{i=0}^{m-1} \theta_{i}^{(0,c)} (-1)^i \frac{1}{n^{c+i}} + O(\frac{1}{n^{c+m}}),\\
r_{0}^{(1,n)} &=&   \sum_{i=0}^{m-1} \theta_{i}^{(1,c)} (-1)^i \frac{1}{n^{c+i}} + O(\frac{1}{n^{c+m}}).
\end{eqnarray*}

Attracting the coefficient of $1/n^{c+m}$ in (\ref{eq:B3}) and (\ref{eq:B4}) and arranging the result yields 
\begin{eqnarray*}
\theta_{m}^{(0,c)} :=
\frac{\lambda}{\mu} \theta_{m-1}^{(0,c)}  + \frac{\nu}{\mu}\theta_{m}^{(0,c-1)}, \\
\theta_{m}^{(1,c)} :=
\frac{\lambda}{\mu} \theta_{m-1}^{(1,c)} + \frac{\nu}{\mu}\theta_{m}^{(1,c-1)}. 
\end{eqnarray*}
Thus, we obtain the desired result for the case $k=c$. 

$\bullet$ {\bf The case $k=0$}

We can prove Lemma~\ref{lem:3-2} using the same methodology. Equations (\ref{eq:theo3.1}) and (\ref{eq:theo3.2}) with $k=1$ 
and Lemma~\ref{lem:3-2} yield
\begin{eqnarray*}
r_{c-1}^{(0,n)} &=& \sum_{i=0}^{m} \theta_{i}^{(0,1)} (-1)^i \frac{1}{n^{1+i}} + O(\frac{1}{n^{m+2}}), \\
r_{c-1}^{(1,n)} &=& \sum_{i=0}^{m} \theta_{i}^{(1,1)} (-1)^i \frac{1}{n^{1+i}} + O(\frac{1}{n^{m+2}}),  \\
r_{c}^{(0,n)}& =& \sum_{i=0}^{m-1} \theta_{i}^{(0,0)} (-1)^i \frac{1}{n^{i}} + O(\frac{1}{n^{m}}), \\
r_{c}^{(1,n)} &=& \sum_{i=0}^{m-1} \theta_{i}^{(1,0)} (-1)^i \frac{1}{n^{i}} + O(\frac{1}{n^{m}}), 
\\
(n+1)r_{c-1}^{(0,n+1)} &=&   \sum_{i=0}^{m} \theta_{i}^{(0,1)} (-1)^i \frac{1}{(n+1)^{i}} + O(\frac{1}{n^{m+1}}) \\
&=& \sum_{j=0}^{m} \Phi_{j}^{(0,1)} (-1)^i \frac{1}{n^{j}} + O(\frac{1}{n^{m+1}}),\\
(n+1)r_{c-1}^{(1,n+1)} &=& \sum_{i=0}^{m} \theta_{i}^{(1,1)} (-1)^i \frac{1}{(n+1)^{i}} + O(\frac{1}{n^{m+1}})\\
&=& \sum_{j=1}^{m} \widetilde{\Phi}_{j}^{(1,1)} (-1)^i \frac{1}{n^{j}} + O(\frac{1}{n^{m+1}}). \\
\end{eqnarray*}

Using these formulae and attracting the coefficient of $1/n^{m}$ in (\ref{eq:k=0,3}) and (\ref{eq:k=0,4}), we obtain 
\begin{eqnarray*}
\theta_{m}^{(0,0)} &:= & - \frac{\lambda_1}{c\nu} \theta_{m}^{(0,1)} 
+\frac{\mu}{c\nu} \sum_{j=0}^{m} \Phi_{j}^{(0,0)} \theta_{m-j-1}^{(0,1)} (-1)^{j+1}  \\ 
  & & +\frac{\mu}{c\nu} \sum_{j=1}^{m} \widetilde{\Phi}_{j}^{(1,0)} \theta_{m-j}^{(0,0)} (-1)^{j}, 
\\
\theta_{m}^{(1,0)} & := & - \frac{\lambda_1}{c\nu} \theta_{m}^{(1,1)} 
+\frac{\mu}{c\nu} \sum_{j=0}^{m} \Phi_{j}^{(0,0)} \theta_{m-j-1}^{(1,1)} (-1)^{j+1} \\
& &  \mbox{} +\frac{\mu}{c\nu} \sum_{j=1}^{m} \widetilde{\Phi}_{j}^{(1,0)} \theta_{m-j}^{(1,0)} (-1)^{j}. 
\end{eqnarray*}
\end{proof}


\section{Proof of Lemma 5.1}\label{app:D} 
\begin{proof}
Let 
\[
\vc{U}_{k}^{(n)} = \vc{Q}_1^{(n)} + \vc{R}_{k-1}^{(n+1)} \vc{Q}_{2}^{(n+1)} ,\ \ n,k \in \mathbb{N}.  
\]
From~\ref{prop:1}, we have $\vc{R}_{k}^{(n)} \vc{U}_{k}^{(n)} = - \vc{Q}_{0}^{(n-1)},\ \ n,k \in \mathbb{N}$. Because the first $c-1$ rows in both sides are zeros, we obtain 
\begin{eqnarray}
\left( 
\begin{array}{c}
\vc{r}^{(0,n)}_{k} \\
\vc{r}^{(1,n)}_{k} \\
\end{array} 
\right)  \vc{U}_{k}^{(n)}  = 
\left( 
\begin{array}{c}
0,0,\dots, - \lambda_2,\ 0 \\
0,0,\dots,0,-\lambda \\
\end{array} 
\right). \label{eq:U}
\end{eqnarray} 
Since {\rm rank}($ \vc{U}_{k}^{(n)} $) $=c+1$, $\vc{r}^{(0,n)}_{k}$ and $\vc{r}^{(1,n)}_{k} $ are uniquely determined. For simplicity, let $\vc{r}^{(0,n)}_{k} = (x_0,x_1,\dots,x_c)$ and $\vc{r}^{(1,n)}_{k} = (y_0,y_1,\dots,y_c)$. 
Comparing both sides of (\ref{eq:U}) yields
\begin{eqnarray}
b_{0}^{(n)}x_0 + \nu x_1& = & 0, \quad i=0 \label{eq:rk00},\  \qquad \\
\lambda x_{i-1} + b_{i}^{(n)} x_{i} +(i+1)\nu  x_{i+1} + \tilde{x}_{i} & = & 0,  \label{eq:rk01}  \\ 
                            i=1,2,\ldots,c-2                                  &     & \nonumber \\
\lambda x_{c-2} + b_{c-1}^{(n)} x_{c-1} +c\nu  x_{c} +  \tilde{x}_{c-1} & = & -\lambda_{2}, \label{eq:rk02} \\ 
              i =c-1  && \nonumber \\ 
\lambda_{1}x_{c-1}+ b_{c}^{(n)}x_{c}+ \tilde{x}_{c} & = & 0,\quad i=c \label{eq:rk03},
\end{eqnarray}
where
$\tilde{x}_{i} =(n+1)\mu \left( x_{c-1} r_{k-1,i-1}^{(0,n+1)}+x_{c} r_{k-1,i-1}^{(1,n+1)}\right)$. 
Furthermore, 
\begin{eqnarray*}
b_{0}^{(n)}y_{0} + \nu y_{1}&=&0, \qquad i=0  \label{eq:rk10}  ,\ \\
\lambda y_{i-1} + b_{i}^{(n)} y_{i} +(i+1)\nu  y_{i+1} + \tilde{y}_{i}&=&0, \label{eq:rk11} \\ 
 i=1,2,\ldots,c-1 &&  \\
\lambda_{1}y_{c-1}+ b_{c}^{(n)} y_{c}+ \tilde{y}_{c}&=& -\lambda, \qquad i=c, \label{eq:rk12} 
\end{eqnarray*}
where 
$\tilde{y}_{i} =(n+1)\mu \left( y_{c-1}r_{k-1,i-1}^{(0,n+1)}+y_{c}r_{k-1,i-1}^{(1,n+1)}\right)$.\\
For arbitrary $n$ and $k$, we express $x_{i}$ as follows.
\begin{eqnarray*}
x_{i} = \alpha_{i} + \beta_{i} x_{c},\ \ i=0,1,\dots,c.
\end{eqnarray*} 

It is obvious that for $i=c$ we have $\alpha_{c}=0$ and $\beta_{c}=1$. For the case $i=c-1$, substituting $x_{c-1}^{(0,n)} = \alpha_{c-1} + \beta_{c-1} x_{c}$ into (\ref{eq:rk03}) yields
\begin{eqnarray*}
\lefteqn{- b_{c}^{(n)} x_{c}  = \lambda_1 ( \alpha_{c-1} + \beta_{c-1} x_{c}) } \\ 
                      & & \mbox{} +(n+1)  \left\{ ( \alpha_{c-1} + \beta_{c-1} x_{c}) r_{k-1,c-1}^{(0,n+1)} + x_{c} r_{k-1,c-1}^{(1,n+1)}\right\}.
 \end{eqnarray*} 
The above formula is rewritten as follows. 
\begin{eqnarray*}
0 &=& \left\{\lambda_1+ (n+1)\mu r_{k-1,c-1}^{(0,n+1)} \right\} \alpha_{c-1}, \\
- b_{c}^{(n)}   -  (n+1)   r_{k-1,c-1}^{(1,n+1)} &=& \left\{  \lambda_1   +(n+1) r_{k-1,c-1}^{(0,n+1)}  \right\}  \beta_{c-1}. 
 \end{eqnarray*} 
Therefore, 
\begin{eqnarray*}
\alpha_{c-1} &=&0, \\
\beta_{c-1}  &=& - \frac{b_{c}^{(n)}  +  (n+1)   r_{k-1,c-1}^{(1,n+1)}}{\lambda_1  +(n+1)  r_{k-1,c-1}^{(0,n+1)}}. 
 \end{eqnarray*} 

For the case $i=c-2$, substituting $x_{c-2} = \alpha_{c-2} + \beta_{c-2} x_{c}$ and $x_{c-1} = \alpha_{c-1} + \beta_{c-1} x_{c}$ into (\ref{eq:rk02}) yields
\begin{eqnarray*}
- b_{c-1}^{(n)}( \alpha_{c-1} + \beta_{c-1} x_{c}) &=& \lambda ( \alpha_{c-2} + \beta_{c-2} x_{c}) +c\nu  x_{c}  \\ 
                                                                  & & \mbox{} + \tilde{x}_{c-1} + \lambda_2.  
\end{eqnarray*}

Rewriting this equation we obtain
\[
\lambda \alpha_{c-2} + \lambda_2 =0
\]
and
\begin{eqnarray*}
\lefteqn{\lambda  \beta_{c-2} + b_{c}^{(n)} \beta_{c-1} + c\nu}  \\ 
&& \mbox{} + (n+1)\mu r_{c-2}^{(0,n+1)} \beta_{c-1} +(n+1)\mu r_{c-2}^{(1,n+1)} = 0.
\end{eqnarray*}

Thus, 
\[
\alpha_{c-2} =- \frac{\lambda_2}{\lambda},
\]
and
\begin{eqnarray*}
\lefteqn{ \beta_{c-2} = } \\
 & & - \frac{b_{c-1}^{(n)}\beta_{c-1} +c\nu+(n+1)\mu r_{k-1,c-2}^{(0,n+1)}\beta_{c-1} }{\lambda} \\ 
 & & - \frac{(n+1)\mu r_{k-1,c-2}^{(1,n+1)}}{\lambda}.
\end{eqnarray*}

The case $i=0,1,\dots,c-3$ is also obtained by transforming (\ref{eq:rk01}) using the same manner. 
\begin{eqnarray*}
\alpha_{i-1} & = & - \frac{b_i^{(n)}\alpha_i +(i+1)\nu\alpha_{i+1}}{\lambda}, \qquad i=c-2,c-3, \dots,1, \\
\beta_{i-1} & = & -  \frac{ b_i^{(n)} \beta_i + (i+1)\nu\beta_{i+1}}{\lambda} \\ 
                &   & \mbox{} - \frac{(n+1)\mu r_{k-1,i-1}^{(0,n+1)} \beta_{c-1} +(n+1)\mu r_{k-1,i-1}^{(1,n+1)}}{\lambda}, \\
               &  & i=c-2,c-3,\dots,1.
\end{eqnarray*}

Furthermore, substituting $x_{0} = \alpha_{0} + \beta_{0} x_{c}$ and $x_{1} = \alpha_{1} + \beta_{1} x_{c}$ into (\ref{eq:rk00}) and arranging the result, we obtain
\[
x_{c} =- \frac{b_0^{(n)}\alpha_0+\nu\alpha_1}{b_0^{(n)}\beta_0+\nu\beta_1}.
\]
\end{proof}



\end{document}